\documentclass[a4paper,UKenglish]{lipics-v2018} 

\usepackage{amsfonts}
\usepackage{amsmath}
\usepackage{amssymb}
\usepackage{booktabs}
\usepackage{mathtools}
\usepackage{microtype}
\usepackage{stmaryrd}
\usepackage{thmtools,thm-restate}
\usepackage{wrapfig}
\usepackage{xcolor}

\usepackage{hyperref,cleveref}
\hypersetup{colorlinks,linkcolor=blue,citecolor=blue,urlcolor=blue}

\usepackage{tikz}
\usetikzlibrary{arrows,calc,automata,shapes.misc}
\tikzset{LMC style/.style={>=angle 60,every edge/.append style={thick},every state/.style={thick,minimum size=20,inner sep=0.5}}}
\tikzset{DFA style/.style={>=angle 60,every edge/.append style={thick},every state/.style={thick,rounded rectangle, minimum size=20,inner sep=0.5}}}

\declaretheorem[name=Theorem,style=plain]{ourtheorem}

\declaretheorem[name=Lemma,Refname={Lemma,Lemmas},sibling=ourtheorem]{ourlemma}

\declaretheorem[name=Example,style=definition,qed=\qedsymbol,sibling=ourtheorem]{ourexample}

\newcommand\ConcurArxiv[2]{#2}

\newcommand*\df[1]{\emph{#1}}

\newcommand{\A}{\mathcal{A}}
\newcommand{\B}{\mathcal{B}}
\newcommand{\bis}{\approx}
\newcommand*\cinf{c_{\mathit{inf}}}
\newcommand*\en[1]{\xrightarrow{#1}}
\newcommand{\M}{\mathcal{M}}
\newcommand*\defeq{\coloneqq}

\newcommand*\Ex{\mathrm{Ex}}
\newcommand*\Fpos{F_{\B'}}
\newcommand{\cras}{\mathit{cras}}

\newcommand*\Matpro{M_\mathit{pro}}
\newcommand*\MCpro{\M_\mathit{pro}}
\newcommand*\NN{\mathbb{N}}
\renewcommand*\P{\mathcal{P}}
\renewcommand*\Pr{\mathrm{Pr}} 
\newcommand*\prho{\rho_\mathit{pro}}
\newcommand{\R}{\mathbb{R}}
\newcommand{\res}[1]{#1\mathord{\downarrow}}
\newcommand*\seeall{\bullet}
\newcommand*\smart{\circ}

\newcommand*\T{\mathsf{T}} 
\renewcommand*\t[1]{\overrightarrow{#1}}

\newcommand*\U{\mathcal{U}}

\newcommand*\ourparagraph[1]{\subparagraph{#1}}
\renewcommand*\vec[1]{\mathbf{#1}}

\title{Selective Monitoring}

\author{Radu Grigore}{University of Kent, UK}{}{https://orcid.org/0000-0003-1128-0311}{Work supported by EPSRC grant EP/R012261/1.}
\author{Stefan Kiefer}{University of Oxford, UK}{}{}{Work supported by a Royal Society University Research Fellowship.}
\authorrunning{R. Grigore and S. Kiefer}
\Copyright{Radu Grigore and Stefan Kiefer}
\subjclass{Theory of computation $\rightarrow$  Randomness, geometry and discrete structures}
\keywords{runtime monitoring, probabilistic systems, Markov chains, automata, language equivalence}

\category{}

\ConcurArxiv{\relatedversion{\url{https://arxiv.org/abs/1806.06143}}}{}

\supplement{}

\funding{}

\acknowledgements{}

\EventEditors{Sven Schewe and Lijun Zhang}
\EventNoEds{2}
\EventLongTitle{29th International Conference on Concurrency Theory (CONCUR 2018)}
\EventShortTitle{CONCUR 2018}
\EventAcronym{CONCUR}
\EventYear{2018}
\EventDate{September 4--7, 2018}
\EventLocation{Beijing, China}
\EventLogo{}
\SeriesVolume{118}
\ArticleNo{20}
\nolinenumbers 

\begin{document}
\renewcommand\sectionautorefname{\S\!}
\renewcommand\subsectionautorefname{\S\!}
\renewcommand\subsubsectionautorefname{\S\!}

\pagestyle{plain}
\sloppy

\maketitle

\begin{abstract} 
We study selective monitors for labelled Markov chains.
Monitors observe the outputs that are generated by a Markov chain during its run, with the goal of identifying runs as correct or faulty.
A monitor is selective if it skips observations in order to reduce monitoring overhead.
We are interested in monitors that minimize the expected number of observations.
We establish an undecidability result for selectively monitoring general Markov chains.
On the other hand, we show for non-hidden Markov chains (where any output identifies the state the Markov chain is in) that simple optimal monitors exist and can be computed efficiently, based on DFA language equivalence.
These monitors do not depend on the precise transition probabilities in the Markov chain.
We report on experiments where we compute these monitors for several open-source Java~projects.

\end{abstract}

\section{Introduction} 

Consider an MC (Markov chain) whose transitions are labelled with letters,
and a finite automaton that accepts languages of infinite words.
Computing the probability that the random word emitted by the MC is accepted by the automaton is a classical problem at the heart of probabilistic verification.
A finite prefix may already determine whether the random infinite word is accepted, and computing the probability that such a \emph{deciding} finite prefix is produced is a nontrivial \emph{diagnosability} problem.
The theoretical problem we study in this paper is how to catch deciding prefixes without observing the whole prefix; i.e., we want to minimize the expected number of observations and still catch all deciding prefixes.

\ourparagraph{Motivation.} 
In runtime verification a program sends messages to a monitor, which decides if the program run is faulty.
Usually, runtime verification is turned off in production code because monitoring overhead is prohibitive.
QVM (quality virtual machine) and ARV (adaptive runtime verification) are existing pragmatic solutions to the overhead problem, which perform best-effort monitoring within a specified overhead budget~\cite{qvm,adaptive-rv}.
ARV relies on RVSE (runtime verification with state estimation) to also compute a probability that the program run is faulty~\cite{stoller2011,kalajdzic2013}.
We take the opposite approach: we ask for the smallest overhead achievable without compromising precision at all.


\ourparagraph{Previous Work.}
Before worrying about the performance of a monitor, one might want to check if faults in a given system can be diagnosed at all.
This problem has been studied under the term \emph{diagnosability}, first for non-stochastic finite discrete event systems~\cite{diagnos-discrete-old}, which are labelled transition systems.
It was shown in~\cite{diagnos-discrete} that diagnosability can be checked in polynomial time, although the associated monitors may have exponential size.
Later the notion of diagnosability was extended to stochastic discrete-event systems, which are labelled Markov chains~\cite{diagnos-stochastic}.
Several notions of diagnosability in stochastic systems exist, and some of them have several names, see, e.g., \cite{Sistla11,BertrandHL14} and the references therein.
Bertrand et al.~\cite{BertrandHL14} also compare the notions.
For instance, they show that for one variant of the problem (referred to as A-diagnosability or SS-diagnosability or IF-diagnosability) a previously proposed polynomial-time algorithm 
is incorrect, and prove that this notion of diagnosability is PSPACE-complete.
Indeed, most variants of diagnosability for stochastic systems are PSPACE-complete~\cite{BertrandHL14}, with the notable exception of AA-diagnosability (where the monitor is allowed to diagnose wrongly with arbitrarily small probability), which can be solved in polynomial time~\cite{Bertrand-LATA}.

\ourparagraph{Selective Monitoring.}
In this paper, we seem to make the problem harder: since observations by a monitor come with a performance overhead, we allow the monitor to skip observations.
In order to decide how many observations to skip, the monitor employs an \emph{observation policy}.
Skipping observations might decrease the probability of deciding (whether the current run of the system is faulty or correct).
We do not study this tradeoff: we require policies to be \emph{feasible}, i.e., the probability of deciding must be as high as under the policy that observes everything.
We do not require the system to be diagnosable; i.e., the probability of deciding may be less than~$1$.
Checking whether the system is diagnosable is PSPACE-complete (\cite{BertrandHL14}, \autoref{thm-diagnosability-PSPACE}).

\ourparagraph{The Cost of Decision in General Markov Chains.}
The \emph{cost} (of decision) is the number of observations that the policy makes during a run of the system.
We are interested in minimizing the expected cost among all feasible policies.
We show that if the system is diagnosable then there exists a policy with finite expected cost, i.e., the policy may stop observing after finite expected time.
(The converse is not true.)
Whether the infimum cost (among feasible policies) is finite is also PSPACE-complete (\autoref{thm-finitary-PSPACE}).
Whether there is a feasible policy whose expected cost is smaller than a given threshold is undecidable (\autoref{thm-undecidable}), even for diagnosable systems.

\ourparagraph{Non-Hidden Markov Chains.}
We identify a class of MCs, namely non-hidden MCs, where the picture is much brighter.
An MC is called \emph{non-hidden} when each label identifies the state.
Non-hidden MCs are always diagnosable.
Moreover, we show that \emph{maximally procrastinating} policies are (almost) optimal (\autoref{thm-max-pro-optimal}).
A policy is called maximally procrastinating when it skips observations up to the point where one further skip would put a decision on the current run in question.
We also show that one can construct an (almost) optimal maximally procrastinating policy in polynomial time.
This policy \emph{does not depend} on the exact probabilities in the MC, although the expected cost under that policy does.
That is, we efficiently construct a policy that is (almost) optimal regardless of the transition probabilities on the MC transitions.
We also show that the infimum cost (among all feasible policies) can be computed in polynomial time (\autoref{thm-non-hidden-cinf}).
Underlying these results is a theory based on automata, in particular, checking language equivalence of DFAs.

\ourparagraph{Experiments.}

We evaluated the algorithms presented in this paper by implementing them in Facebook Infer, and trying them on $11$ of the most forked Java projects on GitHub.
We found that, on average, selective monitoring can reduce the number of observations to a half.

\section{Preliminaries} 

Let $S$ be a finite set.
We view elements of $\R^S$ as \df{vectors}, more specifically as row vectors.
We write $\vec{1}$ for the all-1 vector, i.e., the element of $\{1\}^S$.
For a vector $\mu \in \R^S$, we denote by $\mu^\T$ its transpose, a column vector.
A vector $\mu \in [0,1]^S$ is a \df{distribution} 
\df{over~$S$} if $\mu \vec{1}^\T = 1$. 
For $s \in S$ we write $e_s$ for the (\df{Dirac}) distribution over~$S$ with $e_s(s) = 1$ and $e_s(t) = 0$ for $t \in S \setminus\{s\}$.
We view elements of $\R^{S \times S}$ as \df{matrices}.
A matrix $M \in [0,1]^{S \times S}$ is called \df{stochastic} if each row sums up to one, i.e.,
 $M \vec{1}^\T = \vec{1}^\T$.

For a finite alphabet~$\Sigma$, we write $\Sigma^*$ and~$\Sigma^\omega$ for the finite and infinite words over~$\Sigma$, respectively.
We write $\varepsilon$ for the empty word.
We represent languages $L \subseteq \Sigma^\omega$ using deterministic finite automata, and we represent probability measures $\Pr$ over~$\Sigma^\omega$ using Markov chains.

A (discrete-time, finite-state, labelled) \df{Markov chain (MC)} is a quadruple $(S, \Sigma, M, s_0)$ where $S$ is a finite set of states, $\Sigma$ a finite alphabet, $s_0$ an initial state, and $M : \Sigma \to [0,1]^{S \times S}$ specifies the transitions, such that $\sum_{a \in \Sigma} M(a)$ is a stochastic matrix.
Intuitively, if the MC is in state~$s$, then with probability~$M(a)(s,s')$ it emits~$a$ and moves to state~$s'$.
For the complexity results in this paper, we assume that all numbers in the matrices~$M(a)$ for $a \in \Sigma$ are rationals given as fractions of integers represented in binary.
We extend~$M$ to the mapping $M: \Sigma^* \to [0,1]^{S \times S}$ with $M(a_1 \cdots a_k) = M(a_1) \cdots M(a_k)$ for $a_1, \ldots, a_k \in \Sigma$.
Intuitively, if the MC is in state~$s$ then with probability~$M(u)(s,s')$ it emits the word~$u \in \Sigma^*$ and moves (in $|u|$ steps) to state~$s'$.
An MC is called \df{non-hidden} if for each $a \in \Sigma$ all non-zero entries of~$M(a)$ are in the same column.
Intuitively, in a non-hidden MC, the emitted letter identifies the next state.
An MC $(S, \Sigma, M, s_0)$ defines the standard probability measure~$\Pr$ over~$\Sigma^\omega$, uniquely defined
 by assigning probabilities to cylinder sets $\{u\}\Sigma^\omega$, with $u\in\Sigma^*$, as follows:
\begin{align*}
\Pr(\{u\}\Sigma^\omega) \ \defeq \ e_{s_0} M(u) \vec{1}^\T 
\end{align*}

A \df{deterministic finite automaton (DFA)} is a quintuple $(Q, \Sigma, \delta, q_0, F)$ where $Q$ is a finite set of states, $\Sigma$ a finite alphabet, $\delta: Q \times \Sigma \to Q$ a transition function, $q_0$ an initial state, and $F \subseteq Q$ a set of accepting states.
We extend $\delta$ to $\delta: Q \times \Sigma^* \to Q$ as usual.
A DFA defines a language~$L\subseteq\Sigma^\omega$ as follows:
\begin{align*}
L \ \defeq \ \{\,w\in\Sigma^\omega\mid\text{$\delta(q_0, u) \in F$ for some prefix $u$ of $w$}\,\}
\end{align*}
Note that we do not require accepting states to be visited infinitely often: just once suffices.
Therefore we can and will assume without loss of generality that there is $f$ with $F = \{f\}$ and $\delta(f,a)=f$ for all $a \in \Sigma$.

For the rest of the paper we fix an MC $\M = (S, \Sigma, M, s_0)$ and a DFA $\A = (Q, \Sigma, \delta, q_0, F)$.
We define their composition as the MC $\M \times \A \defeq (S \times Q, \Sigma, M', (s_0, q_0))$ where $M'(a)((s,q),(s',q'))$ equals~$M(a)(s,s')$ if $q' = \delta(q,a)$ and $0$ otherwise.
Thus, $\M$ and $\M \times \A$ induce the same probability measure~$\Pr$.

An \df{observation} $o\in \Sigma_\bot$ is either a letter or the special symbol~$\bot \not\in \Sigma$, which stands for `not~seen'.
An \df{observation policy} $\rho : \Sigma_\bot^*\to\{0,1\}$ is a (not necessarily computable) function that, given the observations made so far, says whether we should observe the next letter.
An observation policy~$\rho$ determines a projection $\pi_\rho : \Sigma^\omega\to\Sigma_\bot^\omega$: we have $\pi_\rho(a_1a_2\ldots\,)=o_1o_2\ldots$ when
\begin{align*}
  o_{n+1} \ &= \ \begin{cases}
    a_{n+1} & \text{if $\rho(o_1\ldots o_n)=1$} \\
    \bot & \text{if $\rho(o_1\ldots o_n)=0$}
  \end{cases}
  &&\text{for all $n\ge 0$}
\end{align*}
We denote the \df{see-all policy} by~$\seeall$; thus, $\pi_\seeall(w)=w$.

In the rest of the paper we reserve
  $a$~for letters,
  $o$~for observations,
  $u$~for finite words,
  $w$~for infinite words,
  $\upsilon$~for finite observation prefixes,
  $s$~for states from an MC,
  and $q$~for states from a DFA.
We write $o_1\sim o_2$ when $o_1$~and~$o_2$ are the same or at least one of them is~$\bot$.
We lift this relation to (finite and infinite) sequences of observations (of the same length).
We write $w \gtrsim \upsilon$ when $u \sim \upsilon$ holds for the length-$|\upsilon|$ prefix $u$ of~$w$.

We say that $\upsilon$ is \df{negatively deciding} when $\Pr(\{w \gtrsim \upsilon \mid w \in L\}) = 0$.
Intuitively, $\upsilon$ is negatively deciding when $\upsilon$ is incompatible (up to a null set) with~$L$.
Similarly, we say that $\upsilon$ is \df{positively deciding} when $\Pr(\{w \gtrsim \upsilon \mid w \not\in L\}) = 0$.
An observation prefix~$\upsilon$ is \df{deciding} when it is positively or negatively deciding.
An observation policy~$\rho$ \df{decides}~$w$ when $\pi_\rho(w)$ has a deciding prefix. 
A \df{monitor} is an interactive algorithm that implements an observation policy:
it processes a stream of letters and, after each letter, it replies with one of `yes', `no', or `skip~$n$ letters', where $n\in \NN\cup\{\infty\}$.

\begin{ourlemma} \label{lem-seeall-most-diagnoser}
For any $w$, if some policy decides~$w$ then $\seeall$ decides~$w$.
\end{ourlemma}
\begin{proof}
Let $\rho$ decide~$w$.
Then there is a deciding prefix $\upsilon$ of~$\pi_\rho(w)$.
Suppose $\upsilon$ is positively deciding, i.e., $\Pr(\{w' \gtrsim \upsilon \mid w' \not\in L\}) = 0$.
Let $u$ be the length-$|\upsilon|$ prefix of~$w$.
Then $\Pr(\{w' \gtrsim u \mid w' \not\in L\}) = 0$, since $\upsilon$ can be obtained from~$u$ by possibly replacing some letters with~$\bot$.
Hence $u$ is also positively deciding.
Since $u$ is a prefix of $w = \pi_\seeall(w)$, we have that $\seeall$ decides~$w$.
The case where $\upsilon$ is negatively deciding is similar.
\end{proof}
It follows that $\max_\rho \Pr(\{w \mid \rho \text{ decides } w\}) \, = \, \Pr(\{w \mid \seeall \text{ decides } w\})$.
We say that a policy~$\rho$ is \df{feasible} when it also attains the maximum, i.e., when
\begin{align*}
 \Pr(\{w \mid \rho \text{ decides } w\}) \ &= \ \Pr(\{w \mid \seeall \text{ decides } w\})\,.
\end{align*}
Equivalently, $\rho$ is feasible when
$
 \Pr(\{w \mid \seeall \text{ decides } w \text{ implies } \rho \text{ decides } w\}) \, = \, 1
$,
i.e., almost all words that are decided by the see-all policy are also decided by~$\rho$.
If $\upsilon=o_1o_2\ldots$ is the shortest prefix of $\pi_\rho(w)$ that is deciding, then the \df{cost of decision~$C_\rho(w)$} is $\sum_{k=0}^{|\upsilon|-1} \rho(o_1\ldots o_k)$.
This paper is about finding feasible observation policies~$\rho$ that minimize~$\Ex(C_\rho)$, the expectation of the cost of decision with respect to~$\Pr$. 

\section{Qualitative Analysis of Observation Policies} 

In this section we study properties of observation policies that are qualitative, i.e., not directly related to the cost of decision.
We focus on properties of observation prefixes that a policy may produce.
 
\paragraph*{Observation Prefixes.}
%
We have already defined deciding observation prefixes.
We now define several other types of prefixes: enabled, confused, very confused, and finitary.
A prefix~$\upsilon$ is \df{enabled} if it occurs with positive probability, $\Pr(\{w \gtrsim \upsilon\})>0$.
Intuitively, the other types of prefixes~$\upsilon$ are defined in terms of what would happen if we were to observe all from now on:
if it is not almost sure that eventually a deciding prefix is reached, then we say $\upsilon$~is confused;
if it is almost sure that a deciding prefix will not be reached, then we say $\upsilon$~is very confused;
if it is almost sure that eventually a deciding or very confused prefix is reached, then we say $\upsilon$~is finitary.
To say this formally, let us make a few notational conventions:
for an observation prefix~$\upsilon$, we write $\Pr(\upsilon)$ as a shorthand for $\Pr(\{\,uw\mid u\sim\upsilon\,\})$;
for a set $\Upsilon$ of observation prefixes, we write $\Pr(\Upsilon)$ as a shorthand for $\Pr\bigl(\bigcup_{\upsilon \in \Upsilon} \{\,uw\mid u\sim\upsilon\,\}\bigr)$.
With these conventions, we define:
\begin{enumerate}
\item $\upsilon$~is \df{confused} when $\Pr(\{\,\upsilon u\mid\text{$\upsilon u$ deciding}\,\})<\Pr(\upsilon)$
\item $\upsilon$~is \df{very confused} when $\Pr(\{\,\upsilon u\mid\text{$\upsilon u$ deciding}\,\}) = 0$
\item $\upsilon$~is \df{finitary} when $\Pr(\{\,\upsilon u\mid\text{$\upsilon u$ deciding or very confused}\,\}) = \Pr(\upsilon)$
\end{enumerate}
Observe that
(a)~confused implies enabled,
(b)~deciding implies not confused, and
(c)~enabled and very confused implies confused.
The following are alternative equivalent definitions:
\begin{enumerate}
\item $\upsilon$~is \df{confused} when $\Pr(\{\,u w \mid u \sim \upsilon,\ \text{no prefix of~$\upsilon w$ is deciding}\,\}) > 0$
\item $\upsilon$~is \df{very confused} when $\upsilon u'$ is non-deciding for all enabled $\upsilon u'$
\item $\upsilon$~is \df{finitary} when $\Pr(\{u w \mid u \sim \upsilon,\ \text{no prefix of~$\upsilon w$ is deciding or very confused}\}) = 0$
\end{enumerate}

\begin{ourexample} \label{ex-obs-prefixes}
Consider the MC and the DFA depicted here:
\begin{center}
\begin{tikzpicture}[scale=2,LMC style]
\node[state] (s0) at (0,0) {$s_0$};
\node[state] (s1) at (-1,0) {$s_1$};
\node[state] (s2) at (1,0) {$s_2$};
\path[->] (0,0.5) edge (s0);
\path[->] (s0) edge[pos=0.4] node[above] {$\frac12 a$} (s1);
\path[->] (s0) edge[pos=0.4] node[above] {$\frac12 a$} (s2);
\path[->] (s1) edge [loop,out=250,in=290,looseness=11] node[pos=0.2,left] {$1 a$} (s1);
\path[->] (s2) edge [loop,out=250,in=290,looseness=11] node[pos=0.2,left] {$\frac12 a$} node[pos=0.8,right] {$\frac12 b$} (s2);
\end{tikzpicture}
\hspace{15mm}
\begin{tikzpicture}[scale=2,LMC style]
\node[state] (q0) at (0,0) {$q_0$};
\node[state,accepting] (q1) at (1,0) {$f$};
\path[->] (0,0.5) edge (q0);
\path[->] (q0) edge [loop,out=250,in=290,looseness=11] node[pos=0.2,left] {$a$} (q0);
\path[->] (q1) edge [loop,out=250,in=290,looseness=11] node[pos=0.2,left] {$a$} node[pos=0.83,right] {$b$} (q1);
\path[->] (q0) edge node[pos=0.4,above] {$b$} (q1);
\end{tikzpicture}
\end{center}
All observation prefixes that do not start with~$b$ are enabled.
The observation prefixes $a b$ and $\bot b$ and, in fact, all observation prefixes that contain~$b$, are positively deciding.
For all $n \in \NN$ we have $\Pr(\{w \gtrsim a^n \mid w \in L\}) > 0$ and $\Pr(\{w \gtrsim a^n \mid w \not\in L\}) > 0$,
so $a^n$~is not deciding.
If the MC takes the right transition first then almost surely it emits~$b$ at some point.
Thus $\Pr(\{a a a \cdots\}) = \frac12$.
Hence $\varepsilon$~is confused.
In this example only non-enabled observation prefixes are very confused.
It follows that $\varepsilon$~is not finitary.
\end{ourexample}

\paragraph*{Beliefs.}
For any~$s$ we write~$\Pr_s$ for the probability measure of the MC~$\M_s$ obtained from~$\M$ by making $s$ the initial state.
For any~$q$ we write~$L_q \subseteq \Sigma^\omega$ for the language of the DFA~$\A_q$ obtained from~$\A$ by making $q$ the initial state.
We call a pair $(s,q)$ \df{negatively deciding} when $\Pr_s(L_q) = 0$;
similarly, we call $(s,q)$ \df{positively deciding} when $\Pr_s(L_q) = 1$.
A subset of $S \times Q$ is called \df{belief}.
We call a belief \df{negatively} (\df{positively}, respectively) \df{deciding} when all its elements are.
We fix the notation $B_0 \defeq \{(s_0, q_0)\}$ (for the \df{initial belief}) for the remainder of the paper.
Define the \df{belief NFA} as the NFA $\B = (S \times Q, \Sigma_\bot, \Delta, B_0 , \emptyset)$ with:
\begin{align*}
\Delta((s,q), a)    &\ = \ \{(s',q') \mid \ M(a)(s,s') > 0, \ \delta(q,a) = q'\} \quad \text{for } a \in \Sigma \\
\Delta((s,q), \bot) &\ = \ \bigcup_{a \in \Sigma} \Delta((s,q), a)
\end{align*}
We extend the transition function $\Delta : (S \times Q) \times \Sigma_\bot \to 2^{S \times Q}$ to  $\Delta : 2^{S \times Q} \times \Sigma_\bot^* \to 2^{S \times Q}$ in the way that is usual for NFAs.
Intuitively, if belief $B$ is the set of states where the product $\M\times\A$ could be now, then $\Delta(B,\upsilon)$ is the belief adjusted by additionally observing~$\upsilon$.
To reason about observation prefixes~$\upsilon$ algorithmically, it will be convenient to reason about the belief $\Delta(B_0, \upsilon)$.

We define confused, very confused, and finitary beliefs as follows:
\begin{enumerate}
\item $B$ is \df{confused} when $\Pr_s(\{\,uw\mid\text{$\Delta(B,u)$ deciding}\,\})<1$ for some $(s,q)\in B$
\item $B$ is \df{very confused} when $\Delta(B,u)$ is empty or not deciding for all~$u$
\item $B$ is \df{finitary} when $\Pr_s(\{\,uw\mid\text{$\Delta(B,u)$ deciding or very confused}\,\})=1$ for all $(s,q)\in B$
\end{enumerate}
\begin{ourexample}
In \autoref{ex-obs-prefixes} we have $B_0 = \{(s_0,q_0)\}$, and $\Delta(B_0,a^n) = \{(s_1,q_0),(s_2,q_0)\}$ for all $n \ge 1$, and $\Delta(B_0,b) = \emptyset$, and $\Delta(B_0,a \bot) = \{(s_1,q_0),(s_2,q_0),(s_2,f)\}$, and $\Delta(B_0, \bot \upsilon) = \{(s_2,f)\}$ for all $\upsilon$ that contain~$b$.
The latter belief $\{(s_2,f)\}$ is positively deciding.
We have $\Pr_{s_1}(\{u w \mid \Delta(\{(s_1,q_0)\},u) \text{ is deciding}\}) = 0$, so any belief that contains $(s_1,q_0)$ is confused.
Also, $B_0$ is confused as $\Pr_{s_0}(\{u w \mid \Delta(\{(s_0,q_0)\},u) \text{ is deciding}\}) = \frac12$.
\end{ourexample}

\paragraph*{Relation Between Observation Prefixes and Beliefs.}
By the following lemma, the corresponding properties of observation prefixes and beliefs are closely related.
\begin{restatable}{ourlemma}{lembeliefbasic} \label{lem-belief-basic}
Let $\upsilon$ be an observation prefix.
\begin{enumerate}
\item $\upsilon$ is enabled if and only if $\Delta(B_0, \upsilon) \ne \emptyset$.
\item $\upsilon$ is negatively deciding if and only if $\Delta(B_0, \upsilon)$ is negatively deciding.
\item $\upsilon$ is positively deciding if and only if $\Delta(B_0, \upsilon)$ is positively deciding.
\item $\upsilon$ is confused if and only if $\Delta(B_0, \upsilon)$ is confused.
\item $\upsilon$ is very confused if and only if $\Delta(B_0, \upsilon)$ is very confused.
\item $\upsilon$ is finitary if and only if $\Delta(B_0, \upsilon)$ is finitary.
\end{enumerate}
\end{restatable}
%
The following lemma gives complexity bounds for computing these properties.%
\begin{restatable}{ourlemma}{lemcomputingtheproperties} \label{lem-computing-the properties}
Let $\upsilon$ be an observation prefix, and $B$ a belief.
\begin{enumerate}
\item Whether $\upsilon$ is enabled can be decided in~P. 
\item Whether $\upsilon$ (or $B$) is negatively deciding can be decided in~P. 
\item Whether $\upsilon$ (or $B$) is positively deciding can be decided in~P. 
\item Whether $\upsilon$ (or $B$) is confused can be decided in~PSPACE.
\item Whether $\upsilon$ (or $B$) is very confused can be decided in~PSPACE.
\item Whether $\upsilon$ (or $B$) is finitary can be decided in~PSPACE.
\end{enumerate}
\end{restatable}
\begin{proof}[Proof sketch]
The belief NFA~$\B$ and the MC $\M \times \A$ can be computed in polynomial time (even in deterministic logspace).
For items~1--3, there are efficient graph algorithms that search these product structures.
For instance, to show that a given pair $(s_1,q_1)$ is not negatively deciding, it suffices to show that $\B$ has a path from $(s_1,q_1)$ to a state $(s_2,f)$ for some~$s_2$.
This can be checked in polynomial time (even in~NL).

For items~4--6, one searches the (exponential-sized) product of~$\M$ and the \emph{determinization} of~$\B$.
This can be done in PSPACE.
For instance, to show that a given belief~$B$ is confused, it suffices to show that there are $(s_1,q_1) \in B$ and $u_1$ and $s_2$ such that $\M$ has a $u_1$-labelled path from $s_1$ to~$s_2$ such that there do \emph{not} exist $u_2$ and $s_3$ such that $\M$ has a $u_2$-labelled path from $s_2$ to~$s_3$ such that $\Delta(B,u_1 u_2)$ is deciding.
This can be checked in NPSPACE = PSPACE by nondeterministically guessing paths in the product of $\M$ and the determinization of~$\B$.
\end{proof}

\paragraph*{Diagnosability.}
We call a policy a \df{diagnoser} when it decides almost surely.
\begin{ourexample} \label{ex-no-feasible}
In \autoref{ex-obs-prefixes} a diagnoser does not exist.
Indeed, the policy~$\seeall$ does not decide when the MC takes the left transition, and decides (positively) almost surely when the MC takes the right transition in the first step.
Hence $\Pr(\{w \mid \seeall \text{ decides } w\}) = \Pr(\Sigma^* \{b\} \Sigma^\omega) = \frac12$.
So $\seeall$ is not a diagnoser.
By \autoref{lem-seeall-most-diagnoser}, it follows that there is no diagnoser.
\end{ourexample}
Diagnosability can be characterized by the notion of confusion:
\begin{restatable}{ourproposition}{propdiagnosabilityconfused} \label{prop-diagnosability-confused}
There exists a diagnoser if and only if $\varepsilon$ is not confused.
\end{restatable}%
The following proposition shows that diagnosability is hard to check.
\begin{restatable}[cf.~{\cite[Theorem~6]{BertrandHL14}}]{ourtheorem}{thmdiagnosabilityPSPACE} \label{thm-diagnosability-PSPACE}
Given an MC~$\M$ and a DFA~$\A$, it is PSPACE-complete to check if there exists a diagnoser.
\end{restatable}
\autoref{thm-diagnosability-PSPACE} essentially follows from a result by Bertrand et al.~\cite{BertrandHL14}.
They study several different notions of diagnosability;
one of them (\emph{FA-diagnosability}) is very similar to our notion of diagnosability.
There are several small differences; e.g., their systems are not necessarily products of an MC and a DFA.
Therefore we give a self-contained proof of \autoref{thm-diagnosability-PSPACE}.
\begin{proof}[Proof sketch]
By \autoref{prop-diagnosability-confused} it suffices to show PSPACE-completeness of checking whether $\varepsilon$~is confused.
Membership in~PSPACE follows from \autoref{lem-computing-the properties}.4.
For hardness we reduce from the following problem: given an NFA~$\U$ over~$\Sigma = \{a,b\}$ where all states are initial and accepting, does $\U$ accept all (finite) words?
This problem is PSPACE-complete~\cite[Lemma~6]{Shallit09}.
\end{proof}

\paragraph*{Allowing Confusion.}

We say an observation policy \df{allows confusion} when, with positive probability, it produces an observation prefix $\upsilon \bot$ such that $\upsilon \bot$ is confused but $\upsilon$ is not.

\begin{restatable}{ourproposition}{propfeasiblenoconfusion} \label{prop-feasible-no-confusion}
A feasible observation policy does not allow confusion.
\end{restatable}
Hence, in order to be feasible, a policy must observe when it would get confused otherwise.
In \autoref{sec:nonhidden} we show that in the non-hidden case there is almost a converse of \autoref{prop-feasible-no-confusion}; i.e., in order to be feasible, a policy need not do much more than not allow confusion.



\section{Analyzing the Cost of Decision} 

In this section we study the computational complexity of finding feasible policies that minimize the expected cost of decision.
We focus on the decision version of the problem:
\emph{Is there a feasible policy whose expected cost is smaller than a given threshold?}
Define:
\[
 \cinf \ \defeq \ \inf_{\text{feasible } \rho} \Ex(C_\rho)
\]
Since the see-all policy~$\seeall$ never stops observing, we have $\Pr(C_\seeall=\infty) = 1$, so $\Ex(C_\seeall) = \infty$.
However, once an observation prefix~$\upsilon$ is deciding or very confused, there is no point in continuing observation.
Hence, we define a \df{light see-all} policy~$\smart$, which observes until the observation prefix~$u$ is deciding or very confused; formally, $\smart(\upsilon) = 0$ if and only if $\upsilon$~is deciding or very confused.
It follows from the definition of very confused that the policy~$\smart$ is feasible.
Concerning the cost~$C_\smart$ we have for all~$w$
\begin{equation} \label{eq-C-smart}
 C_\smart(w) \ = \ \sum_{n=0}^\infty \big( 1 - D_n(w) \big)\,,
\end{equation}
where $D_n(w) = 1$ if the length-$n$ prefix of~$w$ is deciding or very confused, and $D_n(w) = 0$ otherwise.
The following results are proved in the \ConcurArxiv{full version of the paper, on \href{https://arxiv.org/abs/1806.06143}{arXiv}}{appendix}:
\begin{restatable}{ourlemma}{lemepsfinimpliesCsmartfin} \label{lem-eps-fin-implies-Csmart-fin}
If $\varepsilon$~is finitary then $\Ex(C_\smart)$ is finite.
\end{restatable}
\vskip-\topskip 
\begin{restatable}{ourlemma}{lemepsnotfinimpliesinf} \label{lem-eps-not-fin-implies-inf}
Let $\rho$ be a feasible observation policy.
If $\Pr(C_\rho < \infty) = 1$ then $\varepsilon$~is finitary.
\end{restatable}
\vskip-\topskip 
\begin{restatable}{ourproposition}{propfinitarychar} \label{prop-finitary-char}
$\cinf$~is finite if and only if $\varepsilon$ is finitary.
\end{restatable}
\vskip-\topskip 
\begin{restatable}{ourproposition}{propdiagnosersfinite} \label{prop-diagnosers-finite}%
If a diagnoser exists then $\cinf$ is finite.
\end{restatable}
\vskip-\topskip 
\begin{restatable}{ourtheorem}{thmfinitaryPSPACE} \label{thm-finitary-PSPACE}
It is PSPACE-complete to check if $\cinf < \infty$.
\end{restatable}
\autoref{lem-eps-fin-implies-Csmart-fin} holds because, in~$\M\times\A$, a bottom strongly connected component is reached in expected finite time.
\autoref{lem-eps-not-fin-implies-inf} says that a kind of converse holds for feasible policies.
\autoref{prop-finitary-char} follows from \Cref{lem-eps-fin-implies-Csmart-fin,lem-eps-not-fin-implies-inf}.
\autoref{prop-diagnosers-finite} follows from \Cref{prop-diagnosability-confused,prop-finitary-char}.
To show \autoref{thm-finitary-PSPACE}, we use \autoref{prop-finitary-char} and adapt the proof of \autoref{thm-diagnosability-PSPACE}.

%
%
The main negative result of the paper is that one cannot compute~$\cinf$:
\begin{restatable}{ourtheorem}{thmundecidable} \label{thm-undecidable}
It is undecidable to check if $\cinf < 3$, even when a diagnoser exists.
\end{restatable}
\begin{proof}[Proof sketch]
By a reduction from the undecidable problem whether a given probabilistic automaton accepts some word with probability $>\frac12$.
The proof is somewhat complicated.
In fact, in the \ConcurArxiv{full version of the paper (\href{https://arxiv.org/abs/1806.06143}{arXiv})}{appendix} we give two versions of the proof: a short incorrect one (with the correct main idea) and a long correct one.
\end{proof}

\section{The Non-Hidden Case} \label{sec:nonhidden} 

Now we turn to positive results.
In the rest of the paper we assume that the MC~$\M$ is non-hidden, i.e., there exists a function $\t{\cdot} : \Sigma \to S$ such that $M(a)(s,s')>0$ implies $s'=\t{a}$.
We extend~$\t\cdot$ to finite words so that $\t{u a} = \t{a}$.
We write $s \en{u} {}$ to indicate that there is~$s'$ with $M(u)(s,s') > 0$.

\begin{ourexample} \label{ex-hidden-1}
Consider the following non-hidden MC and DFA:
\begin{center}
\begin{tikzpicture}[scale=2,LMC style]
\node[state] (s0) at (0,0) {$\t{a}$};
\node[state] (s1) at (-1,0) {$\t{b}$};
\node[state] (s2) at (1,0) {$\t{c}$};
\path[->] (0,0.5) edge (s0);
\path[->] (s0) edge[pos=0.4] node[above] {$\frac12 b$} (s1);
\path[->] (s0) edge[bend left] node[above] {$\frac12 c$} (s2);
\path[->] (s1) edge [loop,out=160,in=200,looseness=11] node[pos=0.5,left] {$1 b$} (s1);
\path[->] (s2) edge[bend left] node[below] {$1 a$} (s0);
\end{tikzpicture}
\hspace{15mm}
\begin{tikzpicture}[scale=2,LMC style]
\node[state] (q0) at (0,0) {$q_0$};
\node[state,accepting] (q1) at (1,0) {$f$};
\path[->] (0,0.5) edge (q0);
\path[->] (q0) edge [loop,out=160,in=200,looseness=11] node[pos=0.5,left] {$a,b$} (q0);
\path[->] (q0) edge node[pos=0.4,above] {$c$} (q1);
\path[->] (q1) edge [loop,out=340,in=20,looseness=11] node[pos=0.5,right] {$\Sigma$} (q1);
\end{tikzpicture}
\end{center}
\vskip-1ex
\begin{equation*}
\begin{array}{r@{}l@{}l@{\qquad}r@{}l@{}l}
B_0
  &\;\defeq\;
  \{(\t{a},q_0)\}
  &\;\phantom{{}={}}\; \qquad
&
B_2
  &\;\defeq\;
  \Delta(B_0,\bot^2)
  &\;=\;
  \{(\t{b},q_0), (\t{a},f)\}
\\
B_1
  &\;\defeq\;
  \Delta(B_0,\bot)
  &\;=\;
  \{(\t{b},q_0), (\t{c},f)\} \qquad
&
B_3
  &\;\defeq\;
  \Delta(B_0,\bot^2 b)
  &\;=\;
  \{(\t{b},q_0), (\t{b},f)\}
\end{array}
\end{equation*}
$B_0$~is the initial belief.
The beliefs $B_0$ and~$B_1$ are not confused: indeed, $\Delta(B_1, b) = \{(\t{b},q_0)\}$ is negatively deciding, and $\Delta(B_1, a) = \{(\t{a},f)\}$ is positively deciding.
The belief~$B_2$ is confused, as there is no $i \in \NN$ for which $\Delta(B_2, b^i)$ is deciding.
Finally, $B_3$~is very confused.
\end{ourexample}
We will show that in the non-hidden case there always exists a diagnoser (\autoref{lem-pro-diagnoser}).
It follows that feasible policies need to decide almost surely and, by \autoref{prop-diagnosers-finite}, that $\cinf$~is finite.
We have seen in \autoref{prop-feasible-no-confusion} that feasible policies do not allow confusion.
In this section we construct policies that procrastinate so much that they avoid confusion just barely.
We will see that such policies have an expected cost that comes arbitrarily close to~$\cinf$.

\paragraph*{Language Equivalence.}

We characterize confusion by language equivalence in a certain DFA.
Consider the belief NFA~$\B$.
In the non-hidden case, if we disallow $\bot$-transitions then $\B$~becomes a DFA~$\B'$.
For~$\B'$ we define a set of accepting states by $\Fpos \defeq \{(s,q) \mid \Pr_s(L_{q}) = 1\}$.
\begin{ourexample} \label{non-hidden-DFA}
For the previous example, a part of the DFA~$\B'$ looks as follows:
\begin{center}
\begin{tikzpicture}[xscale=2.1,yscale=2,DFA style]
\node[state] (s0) at (0,0)  {$(\t{a},q_0)$};
\node[state] (m1) at (-1,0) {$(\t{b},q_0)$};
\node[state,accepting] (s1) at (1,0)  {$(\t{c},f  )$};
\node[state,accepting] (s2) at (2,0)  {$(\t{a},f  )$};
\node[state,accepting] (s3) at (3,0)  {$(\t{b},f  )$};
\path[->] (s0) edge[pos=0.4] node[above] {$b$} (m1);
\path[->] (s0) edge[pos=0.4] node[above] {$c$} (s1);
\path[->] (m1) edge [loop,out=160,in=200,looseness=8] node[pos=0.2,above] {$b$} (m1);
\path[->] (s1) edge[bend left] node[above] {$a$} (s2);
\path[->] (s2) edge[bend left] node[above] {$c$} (s1);
\path[->] (s2) edge[pos=0.4] node[above] {$b$} (s3);
\path[->] (s3) edge [loop,out=20,in=-20,looseness=8] node[pos=0.2,above] {$b$} (s3);
\end{tikzpicture}
\end{center}
States that are unreachable from $(\t{a},q_0)$ are not drawn here.
\end{ourexample}
We associate with each $(s,q)$ the language $L_{s,q} \subseteq \Sigma^*$ that $\B'$~accepts starting from initial state~$(s,q)$.
We call $(s,q), (s',q')$ \df{language equivalent}, denoted by $(s,q) \bis (s',q')$, when $L_{s,q} = L_{s',q'}$.


\begin{ourlemma} \label{lem-bisim-in-P}
One can compute the relation~$\mathord{\bis}$ in polynomial time.
\end{ourlemma}
\begin{proof}
For any $(s,q)$ one can use standard MC algorithms to check in polynomial time if $\Pr_s(L_q) = 1$ (using a graph search in the composition $\M \times \A$, as in the proof of \autoref{lem-computing-the properties}.3).
Language equivalence in the DFA~$\B'$ can be computed in polynomial time by minimization.
\end{proof}%
We call a belief $B \subseteq  S \times Q$ \df{settled} when all $(s,q) \in B$ are language equivalent.
\begin{restatable}{ourlemma}{lemconfusednotbisimilar} \label{lem-confused-not-bisimilar}
A belief $B \subseteq S \times Q$ is confused if and only if there is $a \in \Sigma$ such that $\Delta(B,a)$ is not settled.
\end{restatable}
It follows that one can check in polynomial time whether a given belief is confused.
We generalize this fact in \autoref{lem-check-procr-k} below.
\begin{ourexample} \label{ex-not-settled}
In \autoref{ex-hidden-1} the belief $B_3$ is not settled.
Indeed, from the DFA in \autoref{non-hidden-DFA} we see that $L_{\t{b},q_0} = \emptyset \ne \{b\}^* = L_{\t{b},f}$.
Since $B_3 = \Delta(B_2,b)$, by \autoref{lem-confused-not-bisimilar}, the belief~$B_2$ is confused.
\end{ourexample}

\paragraph*{Procrastination.}

For a belief~$B \subseteq S \times Q$ and $k \in \NN$, if $\Delta(B, \bot^k)$ is confused then so is $\Delta(B, \bot^{k+1})$.
We define:
\[
 \cras(B)
  \;\defeq\;
  \sup \{\,k \in \NN \mid \Delta(B, \bot^k) \text{ is not confused}\,\}
  \;\in\;
  \NN \cup \{-1,\infty\}
\]
We set $\cras(B) \defeq -1$ if $B$ is confused.
We may write $\cras(s,q)$ for $\cras(\{(s,q)\})$.%
\begin{ourexample}
In \autoref{ex-hidden-1} we have $\cras(B_0) = \cras(\t{a},q_0) = 1$ and $\cras(B_1) = 0$ and $\cras(B_2) = \cras(B_3) = -1$ and $\cras(\t{b},q_0) = \cras(\t{a},f) = \infty$.
\end{ourexample}
\begin{ourlemma} \label{lem-check-procr-k}
Given a belief~$B$, one can compute~$\cras(B)$ in polynomial time.
Further, if $\cras(B)$ is finite then $\cras(B) < |S|^2 \cdot |Q|^2$.
\end{ourlemma}
\begin{proof}
Let $k \in \NN$.
By \autoref{lem-confused-not-bisimilar}, $\Delta(B, \bot^k)$ is confused if and only if:
\[ \exists\,a.\,\exists\,(s,q), (t,r) \in \Delta(B, \bot^k) :  s \en{a},\ t \en{a},\ (\t{a},\delta(q,a)) \not\bis (\t{a},\delta(r,a))
\]
This holds if and only if there is $B_2 \subseteq B$ with $|B_2|\le2$ such that:
\[ \exists\,a.\,\exists\,(s,q), (t,r) \in \Delta(B_2, \bot^k) :  s \en{a},\ t \en{a},\ (\t{a},\delta(q,a)) \not\bis (\t{a},\delta(r,a))
\]
%
%
%
%
Let $G$~be the directed graph with nodes in $S \times Q \times S \times Q$ and edges
\[
 ((s,q,t,r) , (s',q',t',r')) \qquad \Longleftrightarrow \qquad \Delta(\{(s,q),(t,r)\},\bot) \ \supseteq \ \{(s',q'),(t',r')\}\,.
\]
Also define the following set of nodes:
\[
 U \ \defeq \ \{(s,q,t,r) \mid \exists\,a: s \en{a},\ t \en{a},\ (\t{a},\delta(q,a)) \not\bis (\t{a},\delta(r,a))\}
\]
By \autoref{lem-bisim-in-P} one can compute~$U$ in polynomial time.
It follows from the argument above that $\Delta(B, \bot^k)$ is confused if and only if there are $(s,q),(t,r) \in B$ such that there is a length-$k$ path in~$G$ from $(s,q,t,r)$ to a node in~$U$.
Let $k \le |S \times Q \times S \times Q|$ be the length of the shortest such path, and set $k \defeq \infty$ if no such path exists.
Then $k$~can be computed in polynomial time by a search of the graph~$G$, and we have $\cras(B) = k-1$.
\end{proof}

\paragraph*{The Procrastination Policy.}

For any belief~$B$ and any observation prefix~$\upsilon$, the language equivalence classes represented in $\Delta(B, \upsilon)$ depend only on~$\upsilon$ and the language equivalence classes in~$B$.
Therefore, when tracking beliefs along observations, we may restrict~$B$ to a single representative of each equivalence class.
We denote this operation by $\res{B}$.
A belief~$B$ is settled if and only if $|\res{B}| \le 1$.

A \df{procrastination policy}~$\prho(K)$ is parameterized with (a large) $K \in \NN$.
Define (and precompute) $k(s,q) \defeq \min\{K, \cras(s,q)\}$ for all $(s,q)$.
We define $\prho(K)$ by the following monitor that implements it:
\begin{enumerate}
\item $i \defeq 0$ 
\item while $(s_i,q_i)$ is not deciding:
 \begin{enumerate}
 \item skip $k(s_i,q_i)$ observations, then observe a letter~$a_i$
 \item $\{(s_{i+1},q_{i+1})\} \defeq \res{\Delta((s_i,q_i), \bot^{k(s_i,q_i)}a_i)}$;
 \item $i \defeq i+1$;
 \end{enumerate}
\item output yes\slash no decision
\end{enumerate}
It follows from the definition of $\cras$ and \autoref{lem-confused-not-bisimilar} that $\res{\Delta((s_i,q_i), \upsilon_i)}$ is indeed a singleton for all~$i$.
We have:
\begin{ourlemma} \label{lem-pro-diagnoser}
For all $K \in \NN$ the procrastination policy $\prho(K)$ is a diagnoser.
\end{ourlemma}
\begin{proof}
For a non-hidden MC~$\M$ and a DFA~$\A$, there is at most one successor for $(s,q)$ on letter~$a$ in the belief NFA~$\B$, for all $s,q,a$.
Then, by \autoref{lem-confused-not-bisimilar}, singleton beliefs are not confused, and in particular the initial belief~$B_0$ is not confused.
By \autoref{lem-belief-basic}.4, $\varepsilon$ is not confused, which means that $\Pr(\{\,u\mid\text{$u$ deciding}\,\})=\Pr(\varepsilon)=1$.
Since almost surely a deciding word~$u$ is produced and since $\Delta(B_0,u)\subseteq\Delta(B_0,\upsilon)$ whenever $u\sim\upsilon$, it follows that eventually an observation prefix~$\upsilon$ is produced such that $\Delta(B_0,\upsilon)$ contains a deciding pair $(s,q)$.
But, as remarked above, $\Delta(B_0,\upsilon)$ is settled, so it is deciding.
%
\end{proof}

\paragraph*{The Procrastination MC~$\MCpro(K)$.}
The policy~$\prho(K)$ produces a (random, almost surely finite) word $a_1 a_2 \cdots a_{n}$ with $n = C_{\prho(K)}$.
Indeed, the observations that $\prho(K)$~makes can be described by an~MC.
Recall that we have previously defined a composition MC $\M \times \A = (S \times Q, \Sigma, M', (s_0, q_0))$.
Now define an MC $\MCpro(K) \defeq (S \times Q, \Sigma \cup \{\$\}, \Matpro(K), (s_0, q_0))$ where $\$ \not\in \Sigma$ is a fresh letter and the transitions are as follows:
when $(s,q)$ is deciding then $\Matpro(K)(\$)\big((s,q),(s,q)\big) \defeq 1$, and when $(s,q)$ is not deciding then
\begin{align*}
 \Matpro(K)(a)\big((s,q),(\t{a},q')\big) \
 &\defeq\  
 \left( M'(\bot)^{k(s,q)} M'(a) \right)\big((s,q),(\t{a},q')\big)\,,
\end{align*}
where the matrix $M'(\bot) \defeq \sum_a M'(a)$ is powered by~$k(s,q)$.
The MC~$\MCpro(K)$ may not be non-hidden, but could be made non-hidden by (i)~collapsing all language equivalent $(s,q_1),(s,q_2)$ in the natural way, and (ii)~redirecting all $\$$-labelled transition to a new state~$\t{\$}$ that has a self-loop.
In the understanding that $\$\$\$\cdots$ indicates `decision made', the probability distribution defined by the MC~$\MCpro(K)$ coincides with the probability distribution on sequences of non-$\bot$ observations made by~$\prho(K)$.
\begin{ourexample}
For \autoref{ex-hidden-1} the MC~$\MCpro(K)$ for $K \ge 1$ is as follows:
\begin{center}
\begin{tikzpicture}[scale=2,LMC style]
\node[state,inner sep=-1] (s0) at (0,0) {\begin{tabular}{c} $(\t{a},q_0)$ \\ $1$ \end{tabular}};
\node[state,inner sep=-1] (s1) at (-1.8,0) {\begin{tabular}{c} $(\t{b},q_0)$ \\ $\infty$ \end{tabular}};
\node[state,inner sep=-1] (s2) at (+1.8,0) {\begin{tabular}{c} $(\t{a},f  )$ \\ $\infty$ \end{tabular}};
\path[->] (0,0.8) edge (s0);
\path[->] (s0) edge[pos=0.4] node[above] {$\frac12 b$} (s1);
\path[->] (s0) edge[pos=0.4] node[above] {$\frac12 a$} (s2);
\path[->] (s1) edge [loop,out=110,in=70,looseness=7] node[pos=0.2,left] {$1 \$$} (s1);
\path[->] (s2) edge [loop,out=110,in=70,looseness=7] node[pos=0.2,left] {$1 \$$} (s2);
\end{tikzpicture}
\end{center}
Here the lower number in a state indicate the $\cras$ number.
The left state is negatively deciding, and the right state is positively deciding.
The policy~$\prho(K)$ skips the first observation and then observes either $b$ or~$a$, each with probability~$\frac12$, each leading to a deciding belief.
\end{ourexample}

\paragraph*{Maximal Procrastination is Optimal.}
The following lemma states, loosely speaking, that when a belief $\{(s,q)\}$ with $\cras(s,q) = \infty$ is reached and $K$~is large, then a single further observation is expected to suffice for a decision.
\begin{restatable}{ourlemma}{lemlastobs} \label{lem-last-obs}
Let $c(K,s,q)$ denote the expected cost of decision under $\prho(K)$ starting in $(s,q)$.
For each $\varepsilon > 0$ there exists $K \in \NN$ such that for all $(s,q)$ with $\cras(s,q) = \infty$ we have $c(K,s,q) \le 1 + \varepsilon$.
\end{restatable}
\begin{proof}[Proof sketch]
The proof is a quantitative version of the proof of \autoref{lem-pro-diagnoser}.
The singleton belief $\{(s,q)\}$ is not confused.
Thus, if $K$ is large then with high probability the belief $B \defeq \Delta(\{(s,q)\}, \bot^K a)$ (for the observed next letter~$a$) contains a deciding pair $(s',q')$.
But if $\cras(s,q) = \infty$ then, by \autoref{lem-confused-not-bisimilar}, $B$ is settled, so if $B$ contains a deciding pair then $B$ is deciding.
\end{proof}
\begin{ourexample} \label{ex-hidden-2}
Consider the following variant of the previous example:
\begin{center}
\begin{tikzpicture}[scale=2,LMC style]
\node[state] (s0) at (0,0) {$\t{a}$};
\node[state] (s1) at (-1,0) {$\t{b}$};
\node[state] (s2) at (1,0) {$\t{c}$};
\path[->] (0,0.5) edge (s0);
\path[->] (s0) edge[pos=0.4] node[above] {$\frac13 b$} (s1);
\path[->] (s0) edge[pos=0.4] node[above] {$\frac13 c$} (s2);
\path[->] (s0) edge [loop,out=250,in=290,looseness=11] node[pos=0.2,left] {$\frac13 a$} (s0);
\path[->] (s1) edge [loop,out=250,in=290,looseness=11] node[pos=0.2,left] {$1 b$} (s1);
\path[->] (s2) edge [loop,out=250,in=290,looseness=11] node[pos=0.2,left] {$1 c$} (s2);
\end{tikzpicture}
\hspace{15mm}
\begin{tikzpicture}[scale=2,LMC style]
\node[state] (q0) at (0,0) {$q_0$};
\node[state,accepting] (q1) at (1,0) {$f$};
\path[->] (0,0.5) edge (q0);
\path[->] (q0) edge [loop,out=250,in=290,looseness=11] node[pos=0.2,left] {$a$} node[pos=0.83,right] {$b$} (q0);
\path[->] (q0) edge node[pos=0.4,above] {$c$} (q1);
\path[->] (q1) edge [loop,out=250,in=290,looseness=11] node[pos=0.2,left] {$\Sigma$} (q1);
\end{tikzpicture}
\end{center}
The MC~$\MCpro(K)$ for $K \ge 0$ is as follows:
\begin{center}
\begin{tikzpicture}[scale=2,LMC style]
\path[use as bounding box] (-2.5, 0.8) rectangle (2.5, -0.8);
\node[state,inner sep=-3] (s0) at (0,0) {\begin{tabular}{c} $(\t{a},q_0)$ \\ $\infty$ \end{tabular}};
\node[state,inner sep=-3] (s1) at (-2,0) {\begin{tabular}{c} $(\t{b},q_0)$ \\ $\infty$ \end{tabular}};
\node[state,inner sep=-3] (s2) at (+2,0) {\begin{tabular}{c} $(\t{c},f  )$ \\ $\infty$ \end{tabular}};
\path[->] (0,0.8) edge (s0);
\path[->] (s0) edge node[above] {$\frac{1 - (\frac13)^{K+1}}{2} b$} (s1);
\path[->] (s0) edge node[above] {$\frac{1 - (\frac13)^{K+1}}{2} c$} (s2);
\path[->] (s1) edge [loop,out=110,in=70,looseness=7] node[pos=0.2,left] {$1 \$$} (s1);
\path[->] (s2) edge [loop,out=110,in=70,looseness=7] node[pos=0.2,left] {$1 \$$} (s2);
\path[->] (s0) edge [loop,out=250,in=290,looseness=7] node[pos=0.2,left] {$(\frac13)^{K+1} a$} (s0);
\end{tikzpicture}
\end{center}
The left state is negatively deciding, and the right state is positively deciding.
We have $c(K,\t{b},q_0) = c(K,\t{c},f) = 0$ and $c(K,\t{a},q_0) = 1 / ({1 - (\frac13)^{K+1}})$.
\end{ourexample}
Now we can prove the main positive result of the paper:
\begin{restatable}{ourtheorem}{thmmaxprooptimal} \label{thm-max-pro-optimal}
For any feasible policy~$\rho$ there is $K \in \NN$ such that:
\[  \Ex(C_{\prho(K)}) \ \le \ \Ex(C_\rho)
\]
\end{restatable}
\newcommand{\up}{\upsilon_\mathit{pro}}
\newcommand{\ur}{\upsilon_\rho}
\newcommand{\lp}{\ell_\mathit{pro}}
\newcommand{\lr}{\ell_\rho}
\newcommand{\Bp}{B_\mathit{pro}}
\newcommand{\Br}{B_\rho}
\begin{proof}[Proof sketch]
Let $\rho$~be a feasible policy.
We choose $K > |S|^2 \cdot |Q|^2$, so, by \autoref{lem-check-procr-k}, $\prho(K)$~coincides with $\prho(\infty)$ until time, say, $n_\infty$ when $\prho(K)$ encounters a pair $(s,q)$ with $\cras(s,q) = \infty$.
(The time~$n_\infty$ may, with positive probability, never come.)
Let us compare $\prho(K)$ with~$\rho$ up to time~$n_\infty$.
For $n \in \{0, \ldots, n_\infty\}$, define $\up(n)$ and $\ur(n)$ as the observation prefixes obtained by $\prho$ and~$\rho$, respectively, after $n$~steps.
Write $\lp(n)$ and $\lr(n)$ for the number of non-$\bot$ observations in $\up(n)$ and $\ur(n)$, respectively.
For beliefs $B, B'$ we write $B \preceq B'$ when for all $(s,q) \in B$ there is $(s',q') \in B'$ with $(s,q) \bis (s',q')$.
One can show by induction that we have for all $n \in \{0, \ldots, n_\infty\}$:
\begin{equation*} 
 \lp(n) \le \lr(n) \quad \text{ and } \quad \big(\Delta(B_0,\up(n)) \preceq \Delta(B_0,\ur(n)) \quad \text{or} \quad \lp(n) < \lr(n)\big)
\end{equation*}
If time~$n_\infty$ does not come then the inequality $\lp(n) \le \lr(n)$ from above suffices.
Similarly, if at time~$n_\infty$ the pair $(s,q)$ is deciding, we are also done.
If after time~$n_\infty$ the procrastination policy~$\prho(K)$ observes at least one more letter then $\rho$~also observes at least one more letter.
By \autoref{lem-last-obs}, one can choose $K$ large so that for $\prho(K)$ one additional observation probably suffices.
If it is the case that $\rho$~almost surely observes only one letter after~$n_\infty$, then $\prho(K)$ also needs only one more observation, since it has observed at time~$n_\infty$.
\end{proof}

It follows that, in order to compute~$\cinf$, it suffices to analyze $\Ex(C_{\prho(K)})$ for large~$K$.
This leads to the following theorem:
\begin{ourtheorem} \label{thm-non-hidden-cinf}
Given a non-hidden MC~$\M$ and a DFA~$\A$, one can compute~$\cinf$ in polynomial time.
\end{ourtheorem}
\begin{proof}
For each $(s,q)$ define $c(K,s,q)$ as in \autoref{lem-last-obs},
and define
$
 c(s,q)  \defeq  \lim_{K \to \infty} c(K,s,q)
$.
By \autoref{lem-last-obs}, for each non-deciding $(s,q)$ with $\cras(s,q) = \infty$ we have $c(s,q) = 1$.
Hence the $c(s,q)$ satisfy the following system of linear equations where some coefficients come from the procrastination MC~$\MCpro(\infty)$:
\begin{align*}
c(s,q) \ & = \  \begin{cases}
              0 & \text{ if } (s,q) \text{ is deciding} \\
              1 & \text{ if } (s,q) \text{ is not deciding and } \cras(s,q) = \infty \\
              1 + c'(s,q) & \text{ otherwise}
             \end{cases} \\
c'(s,q) \ & = \ \displaystyle\sum_a \sum_{q'} \Matpro(\infty)\big((s,q), (\t{a},q') \big) \cdot c(\t{a},q') \qquad \text{if } \cras(s,q) < \infty
\end{align*}
By solving the system one can compute $c(s_0,q_0)$ in polynomial time.
We have:
\begin{equation*}
  \cinf
=
  \inf_{\text{feasible $\rho$}} \Ex(C_\rho)
\stackrel{\text{Thm\ref{thm-max-pro-optimal}}}{=}
  \lim_{K\to\infty} \Ex(C_{\prho(K)})
=
  c(s_0,q_0)
\end{equation*}
Hence one can compute~$\cinf$ in polynomial time.
\end{proof}

\section{Empirical Evaluation of the Expected Optimal Cost} 

We have shown that maximal procrastination is optimal in the non-hidden case (\autoref{thm-max-pro-optimal}).
However, we have not shown \emph{how much} better the optimal policy is than the see-all baseline.
It appears difficult to answer this question analytically, so we address it empirically.
We implemented our algorithms in a fork of the Facebook Infer static analyzer~\cite{infer}, and applied them to $11$~open-source projects, totaling $80$~thousand Java methods.
We found that in $>90\%$ of cases the maximally procrastinating monitor is trivial and thus the optimal cost is~$0$, because Infer decides statically if the property is violated.
In the remaining cases, we found that the optimal cost is roughly half of the see-all cost, but the variance is high.

\ourparagraph{Design.}

Our setting requires a DFA and an MC representing, respectively, a program property and a program.
For this empirical estimation of the expected optimal cost, the DFA is fixed, the MC shape is the symbolic flowgraph of a real program, and the MC probabilities are sampled from Dirichlet distributions.

The DFA represents the following property:
`there are no two calls to $\mathit{next}$ without an intervening call to $\mathit{hasNext}$'.
To understand how the MC shape is extracted from programs, some background is needed.
Infer~\cite{infer,biabduction} is a static analyzer that, for each method, infers several preconditions and, attached to each precondition, a symbolic path.
For a simple example, consider a method whose body is
  `$\mathsf{if}\,(b)\,x.\mathit{next}(); \; \mathsf{if}\,(!b)\,x.\mathit{next}()$'.
Infer would generate two preconditions for it, $b$~and~$\lnot b$.
In each of the two attached symbolic paths, we can see that $\mathit{next}$ is not called twice, which we would not notice with a control flowgraph.
The symbolic paths are inter-procedural.
If a method $f$ calls a method~$g$, then the path of $f$ will link to a path of~$g$ and, moreover, it will pick one of the paths of~$g$ that corresponds to what is currently known at the call site.
For example, if $g(b)$ is called from a state in which $\lnot b$~holds, then Infer will select a path of $g$ compatible with the condition $\lnot b$.

The symbolic paths are finite because abstraction is applied, including across mutually recursive calls.
But, still, multiple vertices of the symbolic path correspond to the same vertex of the control flowgraph.
For example, Infer may go around a for-loop five times before noticing the invariant.
By coalescing those vertices of the symbolic path that correspond to the same vertex of the control flowgraph we obtain an \df{SFG (symbolic flowgraph)}.
We use such SFGs as the skeleton of MCs.
Intuitively, one can think of SFGs as inter-procedural control flowgraphs restricted based on semantic information.
Vertices correspond to locations in the program text, and transitions correspond to method calls or returns.
Transition probabilities should then be interpreted as a form of static branch prediction.
One could learn these probabilities by observing many runs of the program on typical input data, for example by using the Baum--Welch algorithm~\cite{mle-mc}.
Instead, we opt to show that the improvement in expected observation cost is robust over a wide range of possible transition probabilities, which we do by drawing several samples from Dirichlet distributions.
Besides, recall that the (optimal) procrastination policy does not depend on transition probabilities.

Once we have a DFA and an MC we compute their product.
In some cases, it is clear that the product is empty or universal.
These are the cases in which we can give the verdict right away, because no observation is necessary.
We then focus on the non-trivial cases.

For non-trivial $\mathrm{MC}\times\mathrm{DFA}$ products, we compute the expected cost of the light see-all policy $\Ex(C_\smart)$, which observes all letters until a decision is made and then stops.
We can do so by using standard algorithms \cite[Chapter 10.5]{BaierKatoen-ModCheck-book}.
Then, we compute $\MCpro$, which we use to compute the expected observation cost~$\cinf$ of the procrastination policy (\autoref{thm-non-hidden-cinf}).
Recall that in order to compute $\MCpro$, one needs to compute the $\cras$ function, and also to find language equivalence classes.
Thus, computing $\MCpro$ entails computing all the information necessary for implementing a procrastinating monitor.


\ourparagraph{Methodology.}

We selected $11$~Java projects among those that are most forked on GitHub\ConcurArxiv{}{ (\autoref{tbl:proj-description})}.
We ran Infer on each of these projects.
From the inferred specifications, we built SFGs and monitors that employ light see-all policies and maximal procrastination policies.
From these monitors, we computed the respective expected costs, solving the linear systems using Gurobi~\cite{gurobi}.
Our implementation is in a fork of Infer, on GitHub.

\begin{table}[th]\centering\small 
\caption{
  Reduction in expected observation cost, on real-world data.
  Each SFG (symbolic flowgraph) corresponds to one inferred precondition of a method.
  The size of monitors is measured in number of language equivalence classes.
  (LOC = lines of code;
  GAvg = geometric average.)
}\label{tbl:results}
\renewcommand\arraystretch{0.9} 
\begin{tabular}{@{}lrrr@{\quad}rrr@{\quad}rrr@{}}
\toprule
& \multicolumn{3}{c}{Project Size}
& \multicolumn{3}{c}{Monitors}
& \multicolumn{2}{c}{$\cinf/\Ex(C_\smart)$}
\\
  \cmidrule(lr){2-4}
  \cmidrule(lr){5-7}
  \cmidrule(lr){8-9}
  Name
& Methods & SFGs & LOC
& Count & Avg-Size & Max-Size
& Med & GAvg
\\
\midrule
tomcat
  & 26K
  & 52K
  & 946K
  & 343
  & 69
  & 304
  & 0.53
  & 0.50
\\
okhttp
  & 3K
  & 6K
  & 49K
  & 110
  & 263
  & 842
  & 0.46
  & 0.42
\\
dubbo
  & 8K
  & 16K
  & 176K
  & 91
  & 111
  & 385
  & 0.53
  & 0.51
\\
jadx
  & 4K
  & 9K
  & 48K
  & 204
  & 96
  & 615
  & 0.58
  & 0.50
\\
RxJava
  & 12K
  & 45K
  & 192K
  & 83
  & 41
  & 285
  & 0.52
  & 0.53
\\
guava
  & 22K
  & 43K
  & 1218K
  & 1126
  & 134
  & 926
  & 0.41
  & 0.41
\\
clojure
  & 5K
  & 19K
  & 66K
  & 219
  & 120
  & 767
  & 0.44
  & 0.44
\\
AndroidUtilCode
  & 3K
  & 7K
  & 436K
  & 39
  & 89
  & 288
  & 0.66
  & 0.58
\\
leakcanary
  & 1K
  & 1K
  & 11K
  & 12
  & 79
  & 268
  & 0.66
  & 0.59
\\
deeplearning4j
  & 21K
  & 40K
  & 408K
  & 262
  & 51
  & 341
  & 0.58
  & 0.58
\\
fastjson
  & 2K
  & 7K
  & 47K
  & 204
  & 63
  & 597
  & 0.59
  & 0.53
\\
\bottomrule
\end{tabular}
\end{table}

\ourparagraph{Results.}

The results are given in \autoref{tbl:results}.
We first note that the number of monitors is much smaller than the number of methods, by a factor of $10$ or~$100$.
This is because in most cases we are able to determine the answer statically, by analyzing the symbolic paths produced by Infer.
The large factor should not be too surprising: we are considering a fixed property about iterators, not all Java methods use iterators, and, when they do, it is usually easy to tell that they do so correctly.
Still, each project has a few hundred monitors, which handle the cases that are not so obvious.

\begin{figure}
\centering
\includegraphics[width=55mm]{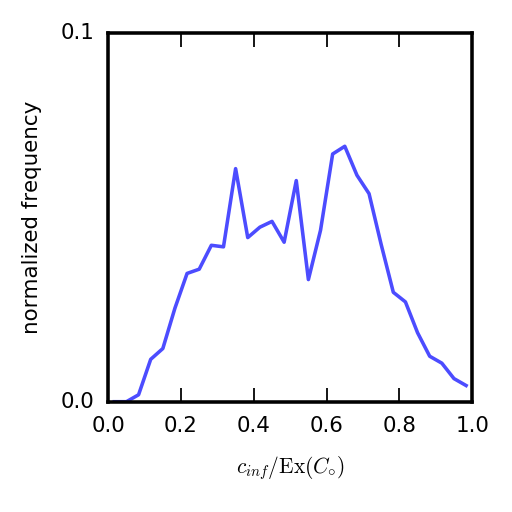}
\caption{Empirical distribution of $\cinf/\Ex(C_\smart)$, across all projects.}
\label{fig:histogram}
\end{figure}


We note that $\frac{\cinf}{\Ex(C_\smart)}\approx 0.5$.
The table supports this by presenting the median and the geometric average, which are close to each-other; the arithmetic average is also close.
There is, however, quite a bit of variation from monitor to monitor, as shown in \autoref{fig:histogram}.
We conclude that selective monitoring has the potential to significantly reduce the overhead of runtime monitoring.

\section{Future Work} 

In this paper we required policies to be feasible, which means that our selective monitors are as precise as non-selective monitors.
One may relax this and study the tradeoff between efficiency (skipping even more observations) and precision (probability of making a decision).
Further, one could replace the diagnosability notion of this paper by other notions from the literature;
one could investigate how to compute $\cinf$ for other classes of MCs, such as acyclic MCs;
one could study the sensitivity of $\cinf$ to changes in transition probabilities;
and one could identify classes of MCs for which selective monitoring helps and classes of MCs for which selective monitoring does not help.

A nontrivial extension to the formal model would be to include some notion of data, which is pervasive in practical specification languages used in runtime verification~\cite{dagstuhl-rv-spec}.
This would entail replacing the DFA with a more expressive device, such as a nominal automaton~\cite{nominal-automata}, a symbolic automaton~\cite{symbolic-automata}, or a logic with data (e.g.,~\cite{ltl-with-freeze}).
Alternatively, one could side-step the problem by using the slicing idea~\cite{slicing}, which separates the concern of handling data at the expense of a mild loss of expressive power.
Finally, the monitors we computed could be used in a runtime verifier, or even in session type monitoring where the setting is similar~\cite{sessiontypes-monitoring}.

\bibliographystyle{plainurl}
\bibliography{selmon}

\ConcurArxiv{}{
  \newpage
  \appendix
  \section{Experimental Details}

\begin{table}[th]\centering 
\caption{Open Source projects, from GitHub.}
\label{tbl:proj-description}
\begin{tabular}{@{}lll@{}}
\toprule
Owner  & Project & Description  \\
\midrule
Apache & tomcat & Java Servlet, Server Pages, and WebSocket \\
Square & okhttp & an HTTP/HTTP2 client for Android \\
Alibaba & dubbo & a high-performance RPC framework \\
Skylot & jadx & Dex to Java decompiler \\
ReactiveX & RxJava & Reactive Extensions for the JVM \\
Google & guava & Google core libraries for Java \\
Clojure & clojure & the Clojure programming language \\
Blankj & AndroidUtilCode & a powerful and easy to use library for Android \\
Square & leakcanary & a memory leak detection library for Android \\
Deeplearning4j & deeplearning4j & Deep Learning for Java on Hadoop and Spark \\
Alibaba & fastjson & a fast JSON parser/generator \\
\bottomrule
\end{tabular}
\end{table}


\section{Proofs}

\subsection{Equivalence of Definitions for Observation Prefixes}\label{sec:def-equiv} 


\ourparagraph{Confused.}
Recall the two definitions for $\upsilon$~confused:
\begin{itemize}
\item[(a)] $\Pr(\{\,\upsilon u\mid \text{$\upsilon u$ deciding}\,\}) < \Pr(\upsilon)$
\item[(b)] $\Pr(\{\,uw\mid\text{$u\sim\upsilon$, no prefix of $\upsilon w$ is deciding}\,\}) > 0$
\end{itemize}
If $A\subseteq B$, then $\Pr(A)<\Pr(B)$ is equivalent to $\Pr(B\setminus A)>0$.
The plan is to pick $A$~and~$B$ such that definition~(a) has the form $\Pr(A)<\Pr(B)$, and then show that $A\subseteq B$ and that definition~(b) has the form $\Pr(B\setminus A)>0$.
\begin{equation*}
\begin{aligned}
\;&\phantom{{}={}}\;
  \Pr(\{\,\upsilon u\mid\text{$\upsilon u$ deciding}\,\})
\\\;&=\;
  \Pr\bigl(\bigcup_{\substack{\upsilon u\\\text{$\upsilon u$ deciding}}} \{\,u'w\mid u'\sim \upsilon u\,\}\bigr)
  &&\text{the $\Pr(\Upsilon)$ notation}
\\\;&=\;
  \Pr(\{\,u'w\mid\text{$u' \sim \upsilon u$ and $\upsilon u$ deciding}\,\})
\\\;&=\;
  \Pr(\{\,uw\mid\text{$u\sim\upsilon$ and $\upsilon w$ has a deciding prefix}\,\})
\end{aligned}
\end{equation*}
Thus, we set
\begin{equation*}
\begin{aligned}
  A \;&\defeq\; \{\,uw\mid\text{$u\sim\upsilon$ and $\upsilon w$ has a deciding prefix}\,\}
\\
  B \;&\defeq\; \{\,uw\mid u\sim \upsilon\,\}
\end{aligned}
\end{equation*}
It is clear that $A\subseteq B$, and that $B\setminus A$ is the event from definition~(b).

\ourparagraph{Very Confused.}
Recall the two definitions for $\upsilon$~very confused:
\begin{itemize}
\item[(a)] $\Pr(\{\,\upsilon u\mid \text{$\upsilon u$ deciding}\,\}) = 0$
\item[(b)] $\upsilon u'$ is non-deciding for all enabled $\upsilon u'$
\end{itemize}
To see that (a)~implies~(b), assume that $\upsilon u'$~is deciding.
Because
$$0 \stackrel{\text{(a)}}= \Pr(\{\,\upsilon u\mid\text{$\upsilon u$ deciding}\,\}) \ge \Pr(\upsilon u')$$
it follows that $\upsilon u'$ is not enabled.

To see that (b)~implies~(a), recall that observation prefixes are finite, and calculate:
\begin{equation*}
\begin{aligned}
\;&\phantom{{}={}}\;
  \Pr(\{\,\upsilon u\mid\text{$\upsilon u$ deciding}\,\})
\\\;&=\;
  \Pr\bigl(\bigcup_{\substack{\upsilon u\\\text{$\upsilon u$ deciding}}}
    \{\,u'w\mid u'\sim \upsilon u\,\}\bigr)
  &&\text{unfold the $\Pr(\Upsilon)$ notation}
\\\;&\le\;
  \sum_{\substack{\upsilon u\\\text{$\upsilon u$ deciding}}}
    \Pr(\{\,u'w\mid u'\sim \upsilon u\,\})
  &&\text{union bound}
\\\;&=\;
  \sum_{\substack{\upsilon u\\\text{$\upsilon u$ deciding}}}
    \Pr(\upsilon u)
  &&\text{fold the $\Pr(\upsilon u)$ notation}
\\\;&=\;
  0
  &&\text{because $\Sigma_\bot^*$ is countable, and we assume (b)}
\end{aligned}
\end{equation*}

\ourparagraph{Finitary.}
Recall the two definitions for $\upsilon$~finitary:
\begin{itemize}
\item[(a)] $\Pr(\{\,\upsilon u\mid \text{$\upsilon u$ very confused or deciding}\,\}) = \Pr(\upsilon)$
\item[(b)] $\Pr(\{u w \mid u \sim \upsilon,\ \text{no prefix of~$\upsilon w$ is deciding or very confused}\}) = 0$
\end{itemize}
The proof is similar to the confused case.
Definition~(a) has the form $\Pr(A)=\Pr(B)$ with $A \subseteq B$, and definition~(b) has the form $\Pr(B\setminus A)=0$.
We just need to make the sets $A$~and~$B$ explicit.
Below, `dv' stands for `deciding or very confused'.
\begin{equation*}
\begin{aligned}
\;&\phantom{{}={}}\;
  \Pr(\{\,\upsilon u\mid \text{$\upsilon u$ deciding or very confused}\,\})
\\\;&=\;
  \Pr\bigl(\bigcup_{\substack{\upsilon u\\\text{$\upsilon u$ dv}}}
    \{\,u'w\mid u'\sim \upsilon u\,\}\bigr)
  &&\text{unfold the $\Pr(\Upsilon)$ notation}
\\\;&=\;
  \Pr(\{\,u' w\mid\text{$u' \sim \upsilon u$ and $\upsilon u$ is dv}\,\})
\\\;&=\;
  \Pr(\{\,uw\mid\text{$u \sim \upsilon$ and $\upsilon w$ has a dv prefix}\,\})
\end{aligned}
\end{equation*}
So, we pick
\begin{equation*}
\begin{aligned}
  A \;&\defeq\; \{\,uw\mid\text{$u\sim \upsilon$ and $\upsilon w$ has a dv prefix}\,\}
\\
  B \;&\defeq\; \{\,uw\mid u\sim \upsilon\,\}
\end{aligned}
\end{equation*}
from which $A \subseteq B$ and $B\setminus A = \{\,uw\mid\text{$u\sim\upsilon$ and $uw$ has no dv prefix}\,\}$, as desired.

\subsection{Proof of \autoref{lem-belief-basic}}

Here is \autoref{lem-belief-basic} from the main body:
\lembeliefbasic*
\begin{proof}
We have:
\begin{equation}
\begin{aligned}
     \Pr(\{w \gtrsim \upsilon\}) \
&= \ \sum_{u \sim \upsilon} \Pr(\{u\} \Sigma^\omega) \\
&= \ \sum_{u \sim \upsilon} \sum_{s} M(u)(s_0,s)  
\end{aligned}
\label{eq-lem-belief-basic-3}
\end{equation}
Similarly we have:
\begin{equation}
\begin{aligned}
     \Pr(\{w \gtrsim \upsilon \mid w \in L\}) \
&= \ \sum_{u \sim \upsilon} \Pr(\{u\} \Sigma^\omega \cap L) \\
&= \ \sum_{u \sim \upsilon} \sum_{s} M(u)(s_0,s)  \cdot \textstyle{\Pr_s}(L_{\delta(q_0, u)})
\end{aligned}
\label{eq-lem-belief-basic-2}
\end{equation}
Further, we have for all $B \subseteq S \times Q$:
\begin{equation}
\begin{aligned}
\Delta(B, \upsilon)
&\ = \ \bigcup_{u \sim \upsilon} \Delta(B,u) \\
&\ = \ \bigcup_{u \sim \upsilon} \{(s',q') \mid \exists\, (s,q) \in B : M(u)(s,s') > 0, \ \delta(q,u) = q'\}
\end{aligned}
\label{eq-lem-belief-basic}
\end{equation}
We prove item~1:
\begin{align*}
                          & \upsilon \text{ is enabled} \\
\Longleftrightarrow \quad & \Pr(\{w \gtrsim \upsilon\}) > 0 && \text{definition} \\
\Longleftrightarrow \quad & \exists\, u \sim \upsilon.\, \exists\, s \text{ with } M(u)(s_0, s)>0  && \text{by~\eqref{eq-lem-belief-basic-3}} \\
\Longleftrightarrow \quad & \exists\, (s,q) \in \Delta(\{(s_0,q_0)\}, \upsilon) && \text{by~\eqref{eq-lem-belief-basic}} \\
\Longleftrightarrow \quad & \Delta(B_0, \upsilon) \ne \emptyset && \text{definition}
\end{align*}
We prove item~2:
\begin{align*}
                    &\quad \upsilon \text{ is negatively deciding} \\
\Longleftrightarrow &\quad \Pr(\{w \gtrsim \upsilon \mid w \in L\}) = 0 && \text{definition}\\
\Longleftrightarrow &\quad \forall\, u \sim \upsilon.\, \forall s \text{ with } M(u)(s_0, s)>0 : \textstyle{\Pr_s}(L_{\delta(q_0, u)})=0 && \text{by~\eqref{eq-lem-belief-basic-2}} \\
\Longleftrightarrow &\quad \forall\, (s,q) \in \Delta(B_0, \upsilon) : \Pr_s(L_q) = 0  && \text{by~\eqref{eq-lem-belief-basic}} \\
\Longleftrightarrow &\quad \text{all } (s,q) \in \Delta(B_0, \upsilon) \text{ are negatively deciding} && \text{definition} \\
\Longleftrightarrow &\quad \Delta(B_0,\upsilon) \text{ is negatively deciding} && \text{definition}
\end{align*}
The proof of item~3 is similar.
Towards item~4 note that if $\upsilon$ is not deciding then none of its prefixes is.
This property explains the second equivalence in the following proof of item~4.
There we write $u' < w$ to denote that $u'$ is a finite prefix of~$w$.
\begin{align*}
                    &\quad \upsilon \text{ is confused} \\
\Longleftrightarrow &\quad \Pr(\{u w \mid u \sim \upsilon,\ \text{no prefix of~$\upsilon w$ is deciding}\}) > 0 && \text{definition}\\
\Longleftrightarrow &\quad \Pr(\{u w \mid u \sim \upsilon,\ \forall\, u'<w : \upsilon u' \text{ is not deciding}\}) > 0 && \text{see above}\\
\Longleftrightarrow &\quad \sum_{u \sim \upsilon} \sum_{s} M(u)(s_0,s) \cdot \mbox{} && \text{as before} \\
                    &\hspace{7mm} \mbox{} \cdot \Pr_s(\{w \mid \forall\, u'<w : \upsilon u' \text{ is not deciding}\}) > 0 \\
\Longleftrightarrow &\quad \exists\, u \sim \upsilon.\, \exists\, s : M(u)(s_0,s) > 0\text{ and} && \text{as before} \\
                    &\hspace{10mm} \mbox{} \Pr_s(\{w \mid \forall\, u'<w : \upsilon u' \text{ is not deciding}\}) > 0 \\
\Longleftrightarrow &\quad \exists\, (s,q) \in \Delta(B_0, \upsilon) : && \text{by~\eqref{eq-lem-belief-basic}} \\
                    &\hspace{10mm} \mbox{} \Pr_s(\{w \mid \forall\, u'<w : \upsilon u' \text{ is not deciding}\}) > 0 \\
\Longleftrightarrow &\quad \exists\, (s,q) \in \Delta(B_0, \upsilon) : && \text{items 2, 3} \\
                    &\hspace{10mm} \mbox{} \Pr_s(\{w \mid \forall\, u'<w : \Delta(\Delta(B_0, \upsilon), u') \text{ is not deciding}\}) > 0 \\
\Longleftrightarrow &\quad \Delta(B_0, \upsilon) \text{ is confused} && \text{definition}
\end{align*}
We prove item~5:
\begin{align*}
                    &\quad \upsilon \text{ is very confused} \\
\Longleftrightarrow &\quad \forall\, u : (\upsilon u \text{ is enabled} \to \upsilon u \text{ is not deciding}) && \text{definition}\\
\Longleftrightarrow &\quad 
                           \forall\, u : (\Delta(B_0, \upsilon u) \ne \emptyset \to \Delta(B_0, \upsilon u) \text{ is not deciding}) && \text{items 2, 3}\\
\Longleftrightarrow &\quad \Delta(B_0, \upsilon) \text{ is very confused} && \text{definition}
\end{align*}
The proof of item~6 is similar to the proof of item~4.
We abbreviate ``not deciding and not very confused'' to ``continuing''.
Note that if $\upsilon$ is continuing then all its prefixes are.
This property explains the second equivalence in the following proof of item~6.
There we write $u' < w$ to denote that $u'$ is a finite prefix of~$w$.
\begin{align*}
                    &\quad \upsilon \text{ is finitary} \\
\Longleftrightarrow &\quad \Pr(\{u w \mid u \sim \upsilon,\ \text{all } u'<\upsilon w \text{ are continuing}\}) = 0 && \text{definition}\\
\Longleftrightarrow &\quad \Pr(\{u w \mid u \sim \upsilon,\ \forall\, u'<w : \upsilon u' \text{ is continuing}\}) = 0 && \text{see above}\\
\Longleftrightarrow &\quad \sum_{u \sim \upsilon} \sum_{s} M(u)(s_0,s) \cdot \mbox{} && \text{as before} \\
                    &\hspace{7mm} \mbox{} \cdot \Pr_s(\{w \mid \forall\, u'<w : \upsilon u' \text{ is continuing}\}) = 0 \\
\Longleftrightarrow &\quad \forall\, u \sim \upsilon.\, \forall\, s : M(u)(s_0,s) = 0\text{ or} && \text{as before} \\
                    &\hspace{10mm} \mbox{} \Pr_s(\{w \mid \forall\, u'<w : \upsilon u' \text{ is continuing}\}) = 0 \\
\Longleftrightarrow &\quad \forall\, (s,q) \in \Delta(B_0, \upsilon) : && \text{by~\eqref{eq-lem-belief-basic}} \\
                    &\hspace{10mm} \mbox{} \Pr_s(\{w \mid \forall\, u'<w : \upsilon u' \text{ is continuing}\}) = 0 \\
\Longleftrightarrow &\quad \forall\, (s,q) \in \Delta(B_0, \upsilon) : && \text{items 2, 3, 5} \\
                    &\hspace{10mm} \mbox{} \Pr_s(\{w \mid \forall\, u'<w : \Delta(B_0, \upsilon u') \text{ is continuing}\}) = 0 \\
\Longleftrightarrow &\quad \forall\, (s,q) \in \Delta(B_0, \upsilon) : && \text{definitions} \\
                    &\hspace{10mm} \mbox{} \Pr_s(\{u' w \mid \Delta(\Delta(B_0, \upsilon), u') \text{ is not continuing}\}) = 1 \\
\Longleftrightarrow &\quad \Delta(B_0, \upsilon) \text{ is finitary} && \text{definition}
\end{align*}
This proves item~6 and completes the proof of the lemma.
\end{proof}

\subsection{Proof of \autoref{lem-computing-the properties}}

Here is \autoref{lem-computing-the properties} from the main body:
\lemcomputingtheproperties*
\begin{proof}
%
Let $G$ be the following graph:
\begin{itemize}
\item the set of vertices is $S \times Q \times P$ where $P$ is the set of prefixes of~$\upsilon$;
\item there is an edge $(s_1, q_1, \upsilon_1) \to (s_2, q_2, \upsilon_2)$ if and only if $\upsilon_1 = o \upsilon_2$ for some~$o$ and $\Delta((s_1,q_1), o) \ni (s_2,q_2)$.
\end{itemize}
This graph~$G$ can be computed in deterministic logspace.
Also note that the belief NFA can be computed in deterministic logspace.

\paragraph*{Proof of item~1.} By \autoref{lem-belief-basic}.1, we have that $\upsilon$ is enabled if and only if $\Delta(B_0, \upsilon) \ne \emptyset$.
We have $\Delta(B_0, \upsilon) \ne \emptyset$ if and only if $G$ has a path from $(s_0, q_0, \upsilon)$ to a node $(s, q, \varepsilon)$.
But graph reachability is in~NL.
This proves item~1.

\paragraph*{Proof of items 2 and~3.}
Assume that, for a given pair $(s_1,q_1)$, one can determine in~NL whether $(s_1,q_1)$ is negatively deciding.
\begin{itemize}
\item Then we can check in~NL whether the belief~$B$ is \emph{not} negatively deciding: guess $(s_1, q_1) \in B$ nondeterministically and check whether $(s_1, q_1)$ is \emph{not} negatively deciding.
    The latter can be done in~NL, as NL is closed under complement.
    It follows that one can check in~NL whether $B$ is negatively deciding.
\item Under the same assumption we can check in~NL whether the observation prefix~$\upsilon$ is negatively deciding:
    By \autoref{lem-belief-basic}.2 we have that $\upsilon$ is not negatively deciding if and only if $\Delta(B_0, \upsilon)$ is not negatively deciding.
    The latter can be checked in~NL by following nondeterministically a path in~$G$ from $(s_0, q_0, \upsilon)$ to a node $(s_1, q_1, \varepsilon)$ and then checking if $(s_1, q_1)$ is not negatively deciding.
    It follows that one can check in~NL whether $\upsilon$ is negatively deciding.
\end{itemize}
The same reasoning applies when ``negatively deciding'' is replaced with ``positively deciding''.
Therefore, for items 2 and~3, it remains to show that one can determine in~NL whether a given pair $(s_1,q_1)$ is negatively (positively, respectively) deciding.
We have:
\begin{align*}
                    &\quad (s_1, q_1) \text{ is not negatively deciding} \\
\Longleftrightarrow &\quad \Pr_{s_1}(L_{q_1}) > 0 \\
\Longleftrightarrow &\quad \exists\, u.\, \exists\, s_2 : M(u)(s_1, s_2) > 0 \text{ and } \delta(q_1, u) = f \\
\Longleftrightarrow &\quad \exists\, u.\, \exists\, s_2 : \Delta(\{(s_1, q_1)\}, u) \ni (s_2, f)
\end{align*}
The latter can be checked in~NL by nondeterministically guessing a word~$u$ letter-by-letter and by nondeterministically following, in the belief NFA, a path that is labelled by~$u$ and leads from $(s_1,q_1)$ to a node $(s_2,f)$.
This proves item~2.

For item~3 it remains to show how to determine in~NL whether a given pair $(s_1, q_1)$ is positively deciding.
We have:
\begin{align*}
                    &\quad (s_1, q_1) \text{ is not positively deciding} \\
\Longleftrightarrow &\quad \Pr_{s_1}(L_{q_1}) < 1 \\
\Longleftrightarrow &\quad \exists\, u.\, \exists\, s_2.\, \exists\, q_2 : M(u)(s_1, s_2) > 0, \ \delta(q_1, u) = q_2, \ \Pr_{s_2}(L_{q_2}) = 0 \\
\Longleftrightarrow &\quad \exists\, u.\, \exists\, (s_2, q_2) \in \Delta(\{(s_1, q_1)\}, u) : (s_2, q_2) \text{ is negatively deciding} \\
\Longleftrightarrow &\quad \exists\, u.\, \exists\, (s_2, q_2) \in \Delta(\{(s_1, q_1)\}, u) : (s_2, q_2) \text{ is negatively deciding,} \\ & \hspace{60mm} |u| \le |S \times Q|
\end{align*}
The last equivalence follows from the fact that there are at most ${|S \times Q|}$ different pairs of the form $(s,q)$.
Therefore, one can check in~NL whether $(s_1, q_1)$ is not positively deciding:
\begin{enumerate}
\item nondeterministically guess a word~$u$ letter-by-letter;
\item nondeterministically follow, in the belief NFA, a path that is labelled by~$u$ and leads from $(s_1,q_1)$ to a node $(s_2,q_2)$;
\item check that $(s_2, q_2)$ is negatively deciding; we have shown previously that this can be done in NL.
\end{enumerate}
This proves item~3.

\paragraph*{Proof of item~4.} By \autoref{lem-belief-basic}.4 it suffices to prove membership in PSPACE for the case where a belief~$B$ is given.
For all $B \subseteq S \times Q$, define:
\[
 V_B \ \defeq \ \{u w \mid \Delta(B,u) \text{ is deciding}\}
\]
We first prove the following claim:

\smallskip\noindent\textbf{Claim.} For given~$s$ and~$B \subseteq S \times Q$ one can determine in~PSPACE whether $\Pr_s(V_B) = 0$.
\smallskip

\noindent We prove the claim.
We have:
\begin{align*}
                    &\quad \Pr_s(V_B) > 0 \\
\Longleftrightarrow &\quad \exists\,u.\, \exists\,s' : M(u)(s,s')>0,\ \Delta(B,u) \text{ is deciding} \\
\Longleftrightarrow &\quad \exists\,u.\, \exists\,s' : M(u)(s,s')>0,\ \Delta(B,u) \text{ is deciding}, \ |u| \le 2^{|S \times Q|}
\end{align*}
The last equivalence follows from the fact that there are at most $2^{|S \times Q|}$ different beliefs of the form $\Delta(B,u)$.
Therefore, one can check in \mbox{NPSPACE} = \mbox{PSPACE} whether $\Pr_s(V_B) > 0$:
\begin{enumerate}
\item guess a word~$u = a_1 \cdots a_n$ with $n \le 2^{|S \times Q|}$ letter-by-letter and check that there are states $s_1, \ldots, s_n$ with $M(a_1)(s,s_1) > 0$ and $M(a_{i+1})(s_i, s_{i+1})>0$ for all $i \in \{1, \ldots, n-1\}$;
\item compute, on the fly, the belief $B' \defeq \Delta(B,u) = \Delta( \cdots \Delta(\Delta(B,a_1),a_2) \cdots a_n)$;
\item check that $B'$ is deciding; we have shown in items 2 and~3 that this can be done in polynomial time.
\end{enumerate}
Since PSPACE is closed under complement, we have proved the claim.

Towards item~4, let $B$ be any belief.
We have:
\begin{align*}
                    &\quad B \text{ is confused} \\ 
\Longleftrightarrow &\quad \exists\,(s,q) \in B : \Pr_s(V_B) < 1 \\
\Longleftrightarrow &\quad \exists\,(s,q) \in B.\,\exists\,u.\, \exists\,s' : M(u)(s,s')>0,\ \Pr_{s'}(V_{\Delta(B,u)}) = 0 \\
\Longleftrightarrow &\quad \exists\,(s,q) \in B.\,\exists\,u.\, \exists\,s' : M(u)(s,s')>0,\ \Pr_{s'}(V_{\Delta(B,u)}) = 0,\\ &\hspace{50mm} |u| \le |S| \cdot 2^{|S \times Q|}
\end{align*}
The last equivalence follows from the fact that there are at most $|S| \cdot 2^{|S \times Q|}$ different pairs of the form $(s',\Delta(B,u))$.
Therefore, one can check in \mbox{NPSPACE} = \mbox{PSPACE} whether $B$ is confused:
\begin{enumerate}
\item guess a word~$u = a_1 \cdots a_n$ with $n \le |S| \cdot 2^{|S \times Q|}$ letter-by-letter and check that there are states $s_1, \ldots, s_n$ with $M(a_1)(s,s_1) > 0$ and $M(a_{i+1})(s_i, s_{i+1})>0$ for all $i \in \{1, \ldots, n-1\}$;
\item compute, on the fly, the belief $B' \defeq \Delta(B,u) = \Delta( \cdots \Delta(\Delta(B,a_1),a_2) \cdots a_n)$;
\item check that $\Pr_{s_n}(V_{B'}) = 0$; we have shown in the claim above that this can be done in~PSPACE.
\end{enumerate}

\paragraph*{Proof of item~5.} By \autoref{lem-belief-basic}.5 it suffices to prove membership in PSPACE for the case where a belief~$B$ is given.
We have:
\begin{align*}
& B \text{ is not very confused} \\
\Longleftrightarrow \quad & \exists\,u : \Delta(B,u) \ne \emptyset \text{ is deciding} && \text{definition} \\
\Longleftrightarrow \quad & \exists\,u : \Delta(B,u) \ne \emptyset \text{ is deciding},\ |u| \le 2^{|S \times Q|} 
\end{align*}
The last equivalence follows from the fact that there are at most $2^{|S \times Q|}$ different beliefs of the form $\Delta(B,u)$.
Therefore, one can check in \mbox{NPSPACE} = \mbox{PSPACE} whether $B$ is not very confused:
\begin{enumerate}
\item guess a word~$u$ with $|u| \le 2^{|S \times Q|}$ letter-by-letter and compute, on the fly, the belief $B' \defeq \Delta(B,u)$;
\item check that $B' \ne \emptyset$ and $B'$ is deciding; we have shown in items 2 and~3 that this can be done in polynomial time.
\end{enumerate}
Since PSPACE is closed under complement, we have proved item~5.

\paragraph*{Proof of item~6.} By \autoref{lem-belief-basic}.6 it suffices to prove membership in PSPACE for the case where a belief~$B$ is given.
The proof is analogous to the one of item~4.
We abbreviate ``deciding or very confused'' to ``dv''.
For all $B \subseteq S \times Q$, define:
\[
 W_B \ \defeq \ \{u w \mid \Delta(B,u) \text{ is dv}\}
\]
We first prove the following claim:

\smallskip\noindent\textbf{Claim.} For given~$s$ and~$B \subseteq S \times Q$ one can determine in~PSPACE whether $\Pr_s(W_B) = 0$.
\smallskip

\noindent We prove the claim.
We have:
\begin{align*}
                    &\quad \Pr_s(W_B) > 0 \\
\Longleftrightarrow &\quad \exists\,u.\, \exists\,s' : M(u)(s,s')>0,\ \Delta(B,u) \text{ is dv} \\
\Longleftrightarrow &\quad \exists\,u.\, \exists\,s' : M(u)(s,s')>0,\ \Delta(B,u) \text{ is dv}, \ |u| \le 2^{|S \times Q|}
\end{align*}
The last equivalence follows from the fact that there are at most $2^{|S \times Q|}$ different beliefs of the form $\Delta(B,u)$.
Therefore, one can check in \mbox{NPSPACE} = \mbox{PSPACE} whether $\Pr_s(W_B) > 0$:
\begin{enumerate}
\item guess a word~$u = a_1 \cdots a_n$ with $n \le 2^{|S \times Q|}$ letter-by-letter and check that there are states $s_1, \ldots, s_n$ with $M(a_1)(s,s_1) > 0$ and $M(a_{i+1})(s_i, s_{i+1})>0$ for all $i \in \{1, \ldots, n-1\}$;
\item compute, on the fly, the belief $B' \defeq \Delta(B,u) = \Delta( \cdots \Delta(\Delta(B,a_1),a_2) \cdots a_n)$;
\item check that $B'$ is dv; we have shown in items 2, 3 and~5 that this can be done in~PSPACE.
\end{enumerate}
Since PSPACE is closed under complement, we have proved the claim.

Towards item~6, let $B$ be any belief.
We have:
\begin{align*}
                    &\quad B \text{ is not finitary} \\ 
\Longleftrightarrow &\quad \exists\,(s,q) \in B : \Pr_s(W_B) < 1 \\
\Longleftrightarrow &\quad \exists\,(s,q) \in B.\,\exists\,u.\, \exists\,s' : M(u)(s,s')>0,\ \Pr_{s'}(W_{\Delta(B,u)}) = 0 \\
\Longleftrightarrow &\quad \exists\,(s,q) \in B.\,\exists\,u.\, \exists\,s' : M(u)(s,s')>0,\ \Pr_{s'}(W_{\Delta(B,u)}) = 0,\\ &\hspace{50mm} |u| \le |S| \cdot 2^{|S \times Q|}
\end{align*}
The last equivalence follows from the fact that there are at most $|S| \cdot 2^{|S \times Q|}$ different pairs of the form $(s',\Delta(B,u))$.
Therefore, one can check in \mbox{NPSPACE} = \mbox{PSPACE} whether $B$ is not finitary:
\begin{enumerate}
\item guess a word~$u = a_1 \cdots a_n$ with $n \le |S| \cdot 2^{|S \times Q|}$ letter-by-letter and check that there are states $s_1, \ldots, s_n$ with $M(a_1)(s,s_1) > 0$ and $M(a_{i+1})(s_i, s_{i+1})>0$ for all $i \in \{1, \ldots, n-1\}$;
\item compute, on the fly, the belief $B' \defeq \Delta(B,u) = \Delta( \cdots \Delta(\Delta(B,a_1),a_2) \cdots a_n)$;
\item check that $\Pr_{s_n}(W_{B'}) = 0$; we have shown in the claim above that this can be done in~PSPACE.
\end{enumerate}
Since PSPACE is closed under complement, we have proved item~6.
\end{proof}

\subsection{Proof of \autoref{prop-diagnosability-confused}}

Here is \autoref{prop-diagnosability-confused} from the main body:
\propdiagnosabilityconfused*
\begin{proof}
We have:
\begin{align*}
                    &\quad \text{there exists a diagnoser} \\ 
\Longleftrightarrow &\quad \text{there exists a policy that decides almost surely} && \text{definition of diagnoser} \\
\Longleftrightarrow &\quad \seeall \text{decides almost surely} && \text{\autoref{lem-seeall-most-diagnoser}} \\
\Longleftrightarrow &\quad \Pr(\{w \mid w \text{ has a deciding prefix}\} \ = \ 1 && \text{definitions} \\
\Longleftrightarrow &\quad \varepsilon \text{ is not confused} && \text{definition of confused}
\end{align*}
\end{proof}

\subsection{Proof of \autoref{thm-diagnosability-PSPACE}}

Here is \autoref{thm-diagnosability-PSPACE} from the main body:
\thmdiagnosabilityPSPACE*
\begin{proof}
By \autoref{prop-diagnosability-confused} it suffices to show PSPACE-completeness of checking whether $\varepsilon$~is confused.
Membership in~PSPACE follows from \autoref{lem-computing-the properties}.4.

For hardness we reduce from the following problem: given an NFA~$\U$ over~$\Sigma = \{a,b\}$ where all states are initial and accepting, does $\U$ accept all (finite) words?
This problem is PSPACE-complete~\cite[Lemma~6]{Shallit09}.
Let $\U = (Q, \Sigma, \delta, Q, Q)$ be the given NFA.
We construct an MC $\M=(Q \cup \{s_0,s_1,s_2\}, \Sigma \cup \{\#\}, M, s_0)$ where $s_0, s_1, s_2 \not\in Q$, and $\# \not\in \Sigma = \{a,b\}$, and the transitions are as follows:
\begin{center}
\begin{tikzpicture}[scale=2,LMC style]
\node[state] (s0) at (0,0) {$s_0$};
\node (Q) at (-1.3,0) [draw,thick,minimum width=40,minimum height=55] {$Q$};
\node[state] (s1) at (+1,0) {$s_1$};
\node[state] (s2) at (-2.6,0) {$s_2$};
\path[->] (0,0.7) edge (s0);
\path[->] (s0) edge[pos=0.5] node[above] {$\frac{1}{2 |Q|} a$} (Q);
\path[->] (s0) edge[pos=0.5] node[above] {$\frac12 a$} (s1);
\path[->] (s1) edge [loop,out=20,in=-20,looseness=15] (s1);
\node at (1.8,0.15) {$\frac12 a$};
\node at (1.8,-0.15) {$\frac12 b$};
\path[->] (Q) edge node[above] {$\frac1{x(q)}\#$} (s2);
\path[->] (s2) edge [loop,out=160,in=200,looseness=15] node[left] {$1\#$} (s2);
\end{tikzpicture}
\end{center}
In this picture, the rectangle labelled with~$Q$ indicates states and transitions that involve the states in~$Q$, i.e., those coming from the NFA~$\U$.
In more detail we define:
\begin{itemize}
\item $M(a)(s_0, q) \defeq \frac{1}{2 |Q|}$ for all $q \in Q$;
\item $x(q) \defeq |\delta(q,a) \cup \delta(q,b)| + 1$ for all $q \in Q$;
\item $M(\sigma)(q,q') \defeq \frac1{x(q)}$ for all $q \in Q$ and both $\sigma \in \{a,b\}$ and all $q' \in \delta(q,\sigma)$;
\item $M(\#)(q,s_2) \defeq \frac1{x(q)}$ for all $q \in Q$.
\end{itemize}
Define $\A$ to be a DFA that accepts $L = \Sigma^* \{\#\} (\Sigma \cup \{\#\})^*$.
We show that $\varepsilon$~is confused if and only if $\U$ accepts all words.

Suppose $\U$ does not accept the word $u \in \Sigma^*$.
Then one of the following two events happens almost surely, depending on whether $\M$ takes the left or the right transition in the first step:
\begin{itemize}
\item $\M$ emits $a u_0 \#$ for some $u_0 \in \Sigma^*$;
\item $\M$ emits $a u_0 u$ for some $u_0 \in \Sigma^*$.
\end{itemize}
We claim that both $a u_0 \#$ and $a u_0 u$ are deciding, implying that $\varepsilon$ is not confused.
It is clear that $a u_0 \#$ is positively deciding.
We argue that $a u_0 u$ is negatively deciding.
Indeed, since $\U$ does not accept~$u$, we have $\delta(q,u) = \emptyset$ for all $q \in Q$.
Hence, starting from any $q \in Q$, the MC~$\M$ cannot emit~$u$.
So if $\M$ emits $a u_0 u$, it must have taken the right transition in the first step and thus will never emit~$\#$, which means that $a u_0 u$ is negatively deciding.

Conversely, suppose that $\varepsilon$ is not confused.
Then, almost surely, $\M$ emits a deciding prefix~$u'$.
By our construction, the word~$u'$ is positively deciding with probability~$\frac12$, and negatively deciding with probability~$\frac12$ (depending on which transition $\M$ takes in the first step).
Hence there exists $u \in \Sigma^*$ such that $a u$ is negatively deciding.
So $\M$ cannot take the left transition in the first step and then emit~$u$.
It follows that there is no $q \in Q$ with $\delta(q,u) \ne \emptyset$, i.e., $\U$ does not accept~$u$.
\end{proof}

\subsection{Proof of \autoref{prop-feasible-no-confusion}}

Here is \autoref{prop-feasible-no-confusion} from the main body:
\propfeasiblenoconfusion*
\begin{proof}
%
%
Suppose that $\rho$~allows confusion, i.e., with positive probability, it produces an observation prefix $\upsilon \bot$ such that $\upsilon \bot$ is confused but $\upsilon$ is not.
We have:
\begin{equation} \label{eq-prop-feasible-no-confusion}
\begin{aligned}
    &\ \Pr(\{w \gtrsim \upsilon \mid \seeall \text{ decides } w \}) \\ 
 =  &\ \Pr(\{u w \mid u \sim \upsilon,\ u w \text{ has a deciding prefix}\}) && \text{definitions} \\
\ge &\ \Pr(\{u w \mid u \sim \upsilon,\ \upsilon w \text{ has a deciding prefix}\}) && \text{} \\
 =  &\ \Pr(\{u w \mid u \sim \upsilon\}) && \text{$\upsilon$ is not confused} \\
 =  &\ \Pr(\{w \gtrsim \upsilon\}) && \text{definition of }\mathord{\gtrsim}
\end{aligned}
\end{equation}
Thus we have:
\begin{align*}
    &\ \Pr(\{w \mid \seeall \text{ decides } w, \ \rho \text{ does not decide } w\}) \\ 
\ge &\ \Pr(\{w \gtrsim \upsilon \mid \seeall \text{ decides } w, \ \rho \text{ does not decide } w\}) \\ 
\ge &\ \Pr(\{w \gtrsim \upsilon \mid \rho \text{ does not decide } w\}) && \text{by~\eqref{eq-prop-feasible-no-confusion}} \\ 
 =  &\ \Pr(\{w \gtrsim \upsilon \bot \mid \rho \text{ does not decide } w\}) && \text{} \\ 
 =  &\ \Pr(\{u w \mid u \sim \upsilon \bot,\ \rho \text{ does not decide } w\}) && \text{definition of }\mathord{\gtrsim} \\ 
 =  &\ \Pr(\{u w \mid u \sim \upsilon \bot,\ \pi_\rho(u w) \text{ has no deciding prefix}\}) && \text{definition of decide} \\ 
\ge &\ \Pr(\{u w \mid u \sim \upsilon \bot,\ \upsilon \bot w \text{ has no deciding prefix}\}) && \text{definition of $\rho$} \\
 >  &\ 0 && \upsilon \bot \text{ is confused}
\end{align*}
Hence $\rho$ is not feasible.
\end{proof}

\subsection{Proof of \autoref{lem-eps-fin-implies-Csmart-fin}}

Here is \autoref{lem-eps-fin-implies-Csmart-fin} from the main body:
\lemepsfinimpliesCsmartfin*
\begin{proof}
We abbreviate ``deciding or very confused'' to ``dv''.
Note that if a belief~$B$ is dv then all $\Delta(B, a)$ are dv.
We use this to show the following similar property:

\smallskip\noindent\textbf{Claim.} For any finitary belief~$B$, all $\Delta(B, a)$ are finitary.
\smallskip

\noindent We prove the claim.
Let $B$ be a finitary belief.
Then we have:
\begin{align*}
1 &\ = \ \Pr_s(\{u w \mid \Delta(B,u) \text{ is dv}\})  \\
  &\ = \ \sum_{a \in \Sigma} \sum_{s'} M(a)(s,s') \cdot \Pr_{s'}(\{u w \mid \Delta(B,a u) \text{ is dv}\})
\end{align*}
Here, the second equality follows from the property of dv mentioned above.
Hence, for all $a \in \Sigma$ and all $(s',q') \in \Delta(B,a)$ we have
$
\Pr_{s'}(\{u w \mid \Delta(B,a u) \text{ is dv}\}) = 1
$,
so $\Delta(B,a)$ is finitary for all~$a$.
This proves the claim.

\newcommand{\Fin}{\mathcal{F}}%
\newcommand{\edv}{e_{\mathit{dv}}}%
Define: 
\[
 \Fin \ \defeq \ \{(s,B) \in S \times 2^{S \times Q} \mid B \text{ finitary}, \ \exists\,q : (s,q) \in B\}
\]
For any belief~$B$ define:
\begin{equation} \label{eq-def-edv}
 \edv(B) \ \defeq \ \begin{cases} 1 &\text{ if $B$ is dv} \\ 
                                0 &\text{ otherwise}
                  \end{cases}
\end{equation}
For any $(s,B) \in \Fin$ and any $n \in \NN$ define
\begin{equation} \label{eq-cn}
 c_n(s,B) \ \defeq \ \Ex_s\big( 1 - D_{B,n} \big) \,,
\end{equation}
where by~$\Ex_s$ we denote the expectation with respect to~$\Pr_s$ and we define $D_{B,n}(w) \defeq \edv(\Delta(B, u_n))$ where $u_n$ is the length-$n$ prefix of~$w$.
We have $c_0(s,B) = 1 - \edv(B)$.
Further, due to the claim above, we have for all $n \in \NN$:
\[
 c_{n+1}(s,B) \ = \ \sum_{a \in \Sigma} \mathop{\sum_{s'}}_{M(a)(s,s')>0} M(a)(s,s') \cdot c_{n}(s',\Delta(B,a)) 
\]
We may write those equations in vector form:
\begin{equation} \label{eq-ExCseeall-cost}
\begin{aligned}
 c_0^\T     \ &= \ \vec{1}^\T - \edv^\T \\
 c_{n+1}^\T \ &= \ U \cdot c_n^\T
\end{aligned}
\end{equation}
where
\begin{itemize}
\item $\edv \in \{0,1\}^{\Fin}$ where $\edv(s,B) = \edv(B)$ as in~\eqref{eq-def-edv};
\item $c_n \in \R^{\Fin}$ where $c_n(s, B)$ is as in~\eqref{eq-cn};
\item $U \in [0,1]^{\Fin \times \Fin}$ where 
 \[
  U((s,B),(s',B')) \ = \ \begin{cases} \displaystyle\mathop{\sum_{a \in \Sigma}}_{\Delta(B,a) = B'} M(a)(s,s') & \text{ if $B$ is not dv} \\
                                       0          & \text{ otherwise.}
                         \end{cases}
 \]
\end{itemize}
The matrix~$U$ satisfies $U \vec{1}^\T = \vec{1}^\T - \edv^\T$.
By a straightforward induction it follows:
\begin{equation} \label{eq-U-powers}
 U^n \vec{1}^\T \ = \ \vec{1}^\T - \sum_{i=0}^{n-1} U^i \edv^\T \qquad \text{for all $n \in \NN$}
\end{equation}
From the definition of finitary, for all $(s,B) \in \Fin$ there is~$u$ such that $\Pr_s(\{u\} \Sigma^\omega) > 0$ and $\Delta(B, u)$ is dv.
It follows that there is $n \in \NN$ such that $\sum_{i=0}^{n-1} U^i \edv^\T > \vec{0}^\T$ where the inequality is strict in all entries.
By~\eqref{eq-U-powers}, we have $U^n \vec{1}^\T < \vec{1}^\T$ where the inequality is strict in all entries.
Hence the spectral radius of~$U$ is less than one, and so the matrix series $\sum_{n=0}^\infty U^n$ converges to a finite matrix, say $U^* \in \R^{\Fin \times \Fin}$.

Suppose $\varepsilon$~is finitary.
By \autoref{lem-belief-basic}.6, the belief~$B_0 = \Delta(B_0, \varepsilon)$ is finitary.
Hence $(s_0, B_0) \in \Fin$.
Further we have:
\begin{align*}
\Ex(C_\smart) \ 
&= \ \sum_{n=0}^\infty \Ex\big( 1 - D_{n} \big) && \text{by~\eqref{eq-C-smart}} \\
&= \ \sum_{n=0}^\infty \Ex_{s_0}\big( 1 - D_{B_0,n} \big) && \text{def.\ of $D_n, D_{B_0,n}$, \autoref{lem-belief-basic}} \\
&= \ \sum_{n=0}^\infty c_n(s_0, B_0) && \text{by~\eqref{eq-cn}} \\
&= \ \sum_{n=0}^\infty \Big( U^n \big(\vec{1}^\T - \edv^\T\big) \Big)(s_0,B_0) && \text{by~\eqref{eq-ExCseeall-cost}} \\
&= \ \Big( U^* \big(\vec{1}^\T - \edv^\T\big) \Big)(s_0,B_0) && \text{definition of~$U^*$} \\
&< \ \infty && \text{$U^*$ is finite}
\end{align*}
This completes the proof.
\end{proof}

\subsection{Proof of \autoref{lem-eps-not-fin-implies-inf}}

Here is \autoref{lem-eps-not-fin-implies-inf} from the main body:
\lemepsnotfinimpliesinf*
\begin{proof}
Let $\rho$ be an observation policy.
Suppose $\varepsilon$~is not finitary and $\Pr(C_\rho < \infty) = 1$.
We show that $\rho$ is not feasible.

Observe that for any~$w$, if $C_\rho(w) < \infty$ then there is a (unique) shortest finite prefix of~$w$, say $\tilde u(w)$, such that $\rho$ never observes a letter after $\tilde u(w)$.
Abbreviating ``deciding or very confused'' to ``dv'', we have:
\begin{align*}
0 \ 
&< \ \Pr(\{w \mid \text{no prefix of~$w$ is dv}\}) && \varepsilon \text{ is not finitary} \\
&= \ \Pr(\{w \mid \text{no prefix of~$w$ is dv}, \ C_\rho(w) < \infty\}) && \Pr(C_\rho < \infty) = 1 \\
&= \ \Pr(\{u w \mid \text{no prefix of~$u w$ is dv}, \ u = \tilde u(u w)\}) && \text{observation above} \\
&\le \ \Pr(\{u w \mid u \text{ is not dv}, \ u = \tilde u(u w)\}) && u \text{ is a prefix of } u w \\
& = \ \sum_{u \text{ is not dv}} \Pr(\{u w \mid u = \tilde u(u w)\}) && \Sigma^* \text{ is countable}
\end{align*}
It follows that there is~$u$ that is (i)~not deciding, (ii)~not very confused, and (iii)~such that $\rho$ never observes a letter after~$u$.
Since $u$~is not very confused, there is~$u'$ such that $u u'$ is enabled and deciding.
Hence $\seeall$ decides all words of the form $u u' w$.
Since $\rho$~never observes a letter after~$u$, we have that $\rho$~does not decide any words of the form $u u' w$.
It follows:
\begin{align*}
      & \Pr(\{w \mid \seeall \text{ decides } w,\  \rho \text{ does not decide } w\}) \\
\ge \ & \Pr(\{u u' w \mid \seeall \text{ decides } u u' w,\  \rho \text{ does not decide } u u' w\}) \\
= \ & \Pr(\{u u' w\}) && \text{as explained above} \\
> \ & 0 && u u' \text{ is enabled}
\end{align*}
Hence $\rho$ is not feasible.
\end{proof}

\subsection{Proof of \autoref{prop-finitary-char}}

Here is \autoref{prop-finitary-char} from the main body:
\propfinitarychar*
\begin{proof}
Suppose $\varepsilon$~is finitary.
By \autoref{lem-eps-fin-implies-Csmart-fin} we then have that $\Ex(C_\smart)$ is finite.
Hence $\cinf \le \Ex(C_\smart)$ is finite.

Conversely, suppose $\varepsilon$~is not finitary.
By \autoref{lem-eps-not-fin-implies-inf} we then have $\Pr(C_\rho = \infty) > 0$ for all feasible observation policies.
Thus $\Ex(C_\rho) = \infty$ holds for all feasible observation policies.
Hence $\cinf = \infty$.
\end{proof}

\subsection{Proof of \autoref{prop-diagnosers-finite}}

Here is \autoref{prop-diagnosers-finite} from the main body:
\propdiagnosersfinite*
\begin{proof}
We have:
\begin{align*}
& \text{a diagnoser exists} \\
\Longleftrightarrow\quad & \varepsilon \text{ is not confused} && \text{Prop.~\ref{prop-diagnosability-confused}} \\
\Longleftrightarrow\quad & \Pr(\{w \mid \text{some prefix of } w \text{ is deciding}\}) \ = \ 1 && \text{definition} \\
\Longrightarrow\quad & \Pr(\{w \mid \text{some prefix of } w \text{ is deciding or very confused}\}) \ = \ 1 && \text{} \\
\Longleftrightarrow\quad & \varepsilon \text{ is finitary} && \text{definition} \\
\Longleftrightarrow\quad & \cinf \text{ is finite} && \text{Prop.~\ref{prop-finitary-char}}
\end{align*}
\end{proof}

\subsection{Proof of \autoref{thm-finitary-PSPACE}}

Here is \autoref{thm-finitary-PSPACE} from the main body:
\thmfinitaryPSPACE*
\begin{proof}
By \autoref{prop-finitary-char} it suffices to show PSPACE-completeness of checking whether $\varepsilon$~is finitary.
Membership in~PSPACE follows from \autoref{lem-computing-the properties}.6.

For hardness we use the same reduction as in the proof of \autoref{thm-diagnosability-PSPACE}.
There we reduce from the following PSPACE-complete problem: given an NFA~$\U$ over~$\Sigma = \{a,b\}$ where all states are initial and accepting, does $\U$ accept all (finite) words?
The reduction produces an MC~$\M$ and a DFA~$\A$.

First we show that in the result of that reduction, no enabled word is very confused.
If a word~$u$ contains the letter~$\#$ then $u$~is negatively deciding, hence not very confused.
Let $u_0 \in \Sigma^*$.
\begin{itemize}
\item If the NFA~$\U$ accepts all words then $a u_0 \#$ is enabled and positively deciding;
\item if $B$~does not accept some word~$u$ then $a u_0 u$ is enabled and negatively deciding.
\end{itemize}
In either case it follows that $a u_0$ is not very confused and also that $\varepsilon$ is not very confused.
We conclude that no enabled word is very confused.

Thus we have:
\begin{align*}
& \varepsilon \text{ is finitary} \\ 
\Longleftrightarrow \quad & \Pr(\{w \mid \text{no prefix of } w \text{ is deciding or very confused}\}) = 0 && \text{definition}\\
\Longleftrightarrow \quad & \Pr(\{w \mid \text{no prefix of } w \text{ is deciding}\}) = 0 && \text{as argued}\\
\Longleftrightarrow \quad & \varepsilon \text{ is not confused} && \text{definition} \\
\Longleftrightarrow \quad & \text{NFA } \U \text{ does not accept all words\,,} && \text{} 
\end{align*}
where the last equivalence was shown in the proof of \autoref{thm-diagnosability-PSPACE}.
It follows that checking if $\cinf$~is finite is PSPACE-complete.
\end{proof}

\subsection{Wrong Proof of \autoref{thm-undecidable}}

Here is \autoref{thm-undecidable} from the main body:
\thmundecidable*
First we give a wrong proof.
This proof contains useful ideas and is similar to and simpler than the correct one, but is flawed.
We point out the flaw.
In \autoref{sub-correct-proof-undecidable} we amend the proof, making it more complicated but correct.
The incorrect proof contains only ideas that feature also in the correct one;
but the correct proof can be read without having read the incorrect one.

\begin{proof}[Wrong proof]
We reduce from the emptiness problem for probabilistic automata.
A \emph{probabilistic automaton (PA)} is a tuple $\P = (S, \Sigma, M, s_0, \eta)$ where
 $S$ is a finite set of states,
 $\Sigma$ is a finite alphabet,
 the mapping $M: \Sigma \to [0,1]^{S \times S}$, where $M(a)$ is stochastic for each $a \in \Sigma$, specifies the transitions,
 $s_0$~is an initial state,
 and $\eta \in [0,1]^S$ is a vector of acceptance probabilities.
Extend~$M$ to $M: \Sigma^* \to [0,1]^{S \times S}$ as in the case of MCs.
In the case of PAs, $M(u)$ is stochastic for each $u \in \Sigma^*$.
For each~$u$ define $\Pr_{\P}(u) := e_{s_0} M(u) \eta^\top$.
The probability~$\Pr_{\P}(u)$ can be interpreted as the probability that $\P$ accepts~$u$.
The \emph{emptiness problem} asks, given a PA~$\P$, whether there is~$u$ such that $\Pr_{\P}(u) > \frac12$.
This problem is undecidable~\cite[p.~190, Theorem~6.17]{Paz71}.

Let $\P = (S_\P, \Sigma_\P, M_\P, s_{0\P}, \eta)$ be the given PA.
We will construct a DFA~$\A$ and an MC~$\M$ over the alphabet $\Sigma \defeq \Sigma_\P \cup \{0,1,?\}$, where $0,1,?$ are fresh letters.
Define $\A$ to be a DFA that accepts $L = 0 (\Sigma \Sigma)^* \Sigma 1 \Sigma^\omega$.
We might characterize~$L$ by saying that $1$ appears on an odd position.

We construct an MC $\M = (S, \Sigma, M, s_0)$ with
\[
 S \defeq \{s_0\} \cup (S_\P \times \{0,1,2,3\}) \,,
\] 
such that $s_0$~is a fresh state.
The MC~$\M$ initially splits randomly into a ``$0$-copy'' (with states in $S \times \{0,2\}$) and a ``$1$-copy'' (with states in $S \times \{1,3\}$), in either case emitting the letter~$0$.
Formally, $M(0)(s_0, (s_{0\P},0)) = M(0)(s_0, (s_{0\P},1)) = \frac12$.
The MC~$\M$ is constructed such that if it goes into the $0$-copy then surely it emits an infinite word that is not in~$L$;
and if it goes into the $1$-copy then almost surely it emits an infinite word in~$L$.
Therefore, for an observation policy it suffices to identify which copy $\M$~has entered.

The transitions in~$\M$ depend on $M_\P$ and~$\eta$.
In \autoref{fig-illustration-undecidable-wrong} we illustrate this dependence with an example, where $\Sigma_\P = \{a,b\}$.
\begin{figure}
\begin{center}
\begin{tikzpicture}[scale=2,LMC style]
\node[state] (s)  at (0.5,0) {$s$};
\coordinate  (sa) at (3.0,0.6);
\coordinate  (sb) at (3.0,-0.6);
\node[state] (t1) at (5.5,0.9) {$t_1$};
\node[state] (t2) at (5.5,0.3) {$t_2$};
\node[state] (t3) at (5.5,-0.3) {$t_3$};
\node[state] (t4) at (5.5,-0.9) {$t_4$};
\path (s) edge node[above] {$a$} (sa);
\path (s) edge node[above] {$b$} (sb);
\path[->] (sa) edge node[above] {$\frac17$} (t1);
\path[->] (sa) edge node[below] {$\frac67$} (t2);
\path[->] (sb) edge node[above] {$\frac15$} (t3);
\path[->] (sb) edge node[below] {$\frac45$} (t4);
\end{tikzpicture} \\
\begin{tikzpicture}[scale=2,LMC style]
\coordinate (0) at (0,0);
\path[->] (0) edge[line width=5, line cap=round] (0,-0.7);
\end{tikzpicture}
\\[-3mm]
\begin{tikzpicture}[scale=2,LMC style]
\node[state] (s0)  at (0.5,0) {$(s,0)$};
\node[state] (s2)  at (3,0) {$(s,2)$};
\path[->] (s0) edge node[above] {$\frac{2 \eta(s)}{3} 0$} node[below] {$\frac{3-2\eta(s)}{3} ?$} (s2);
\node[state] (t10) at (5.5,0.9) {$(t_1, 0)$};
\node[state] (t20) at (5.5,0.3) {$(t_2, 0)$};
\node[state] (t30) at (5.5,-0.3) {$(t_3, 0)$};
\node[state] (t40) at (5.5,-0.9) {$(t_4, 0)$};
\path[->] (s2) edge[bend left=30]  node[pos=0.7,above] {$\frac{1 \cdot 1}{3 \cdot 7} a$} node[pos=0.7,below] {$\frac{1 \cdot 1}{6 \cdot 7} 0$} (t10);
\path[->] (s2) edge[bend left=10]  node[pos=0.7,above] {$\frac{1 \cdot 6}{3 \cdot 7} a$} node[pos=0.7,below] {$\frac{1 \cdot 6}{6 \cdot 7} 0$} (t20);
\path[->] (s2) edge[bend right=10] node[pos=0.7,above] {$\frac{1 \cdot 1}{3 \cdot 5} b$} node[pos=0.7,below] {$\frac{1 \cdot 1}{6 \cdot 5} 0$} (t30);
\path[->] (s2) edge[bend right=30] node[pos=0.7,above] {$\frac{1 \cdot 4}{3 \cdot 5} b$} node[pos=0.7,below] {$\frac{1 \cdot 4}{6 \cdot 5} 0$} (t40);
\end{tikzpicture}
\\
\begin{tikzpicture}[scale=2,LMC style]
\node[state] (s0)  at (0.5,0) {$(s,1)$};
\node[state] (s2)  at (3,0) {$(s,3)$};
\path[->] (s0) edge node[above] {$\frac{2 \eta(s)}{3} 1$} node[below] {$\frac{3-2 \eta(s)}{3} ?$} (s2);
\node[state] (t10) at (5.5,0.9) {$(t_1, 1)$};
\node[state] (t20) at (5.5,0.3) {$(t_2, 1)$};
\node[state] (t30) at (5.5,-0.3) {$(t_3, 1)$};
\node[state] (t40) at (5.5,-0.9) {$(t_4, 1)$};
\path[->] (s2) edge[bend left=30]  node[pos=0.7,above] {$\frac{1 \cdot 1}{3 \cdot 7} a$} node[pos=0.7,below] {$\frac{1 \cdot 1}{6 \cdot 7} 1$} (t10);
\path[->] (s2) edge[bend left=10]  node[pos=0.7,above] {$\frac{1 \cdot 6}{3 \cdot 7} a$} node[pos=0.7,below] {$\frac{1 \cdot 6}{6 \cdot 7} 1$} (t20);
\path[->] (s2) edge[bend right=10] node[pos=0.7,above] {$\frac{1 \cdot 1}{3 \cdot 5} b$} node[pos=0.7,below] {$\frac{1 \cdot 1}{6 \cdot 5} 1$} (t30);
\path[->] (s2) edge[bend right=30] node[pos=0.7,above] {$\frac{1 \cdot 4}{3 \cdot 5} b$} node[pos=0.7,below] {$\frac{1 \cdot 4}{6 \cdot 5} 1$} (t40);
\end{tikzpicture}
\vspace{-2mm}
\end{center}
\caption{Illustration of the wrong reduction. The PA at the top has a state~$s$ whose outgoing transitions are such that $M(a)(s,t_1) = \frac17$, $M(a)(s,t_2) = \frac67$, $M(b)(s,t_3) = \frac15$, $M(b)(s,t_4) = \frac45$.
The resulting MC has two corresponding states, $(s,0)$ and $(s,1)$.
From such a state $(s,i)$ it first emits, with a probability depending on $\eta(s)$, either the letter~$i$ (thus giving away in which copy the MC is) or the letter~$?$, and then emits, with probability~$\frac13$, the letter~$i$ (thus giving away in which copy the MC is), or, with probability $\frac23 \cdot \frac{1}{|\Sigma_\P|} = \frac13$ each, the letters $a$ or~$b$.
}
\label{fig-illustration-undecidable-wrong}
\end{figure}
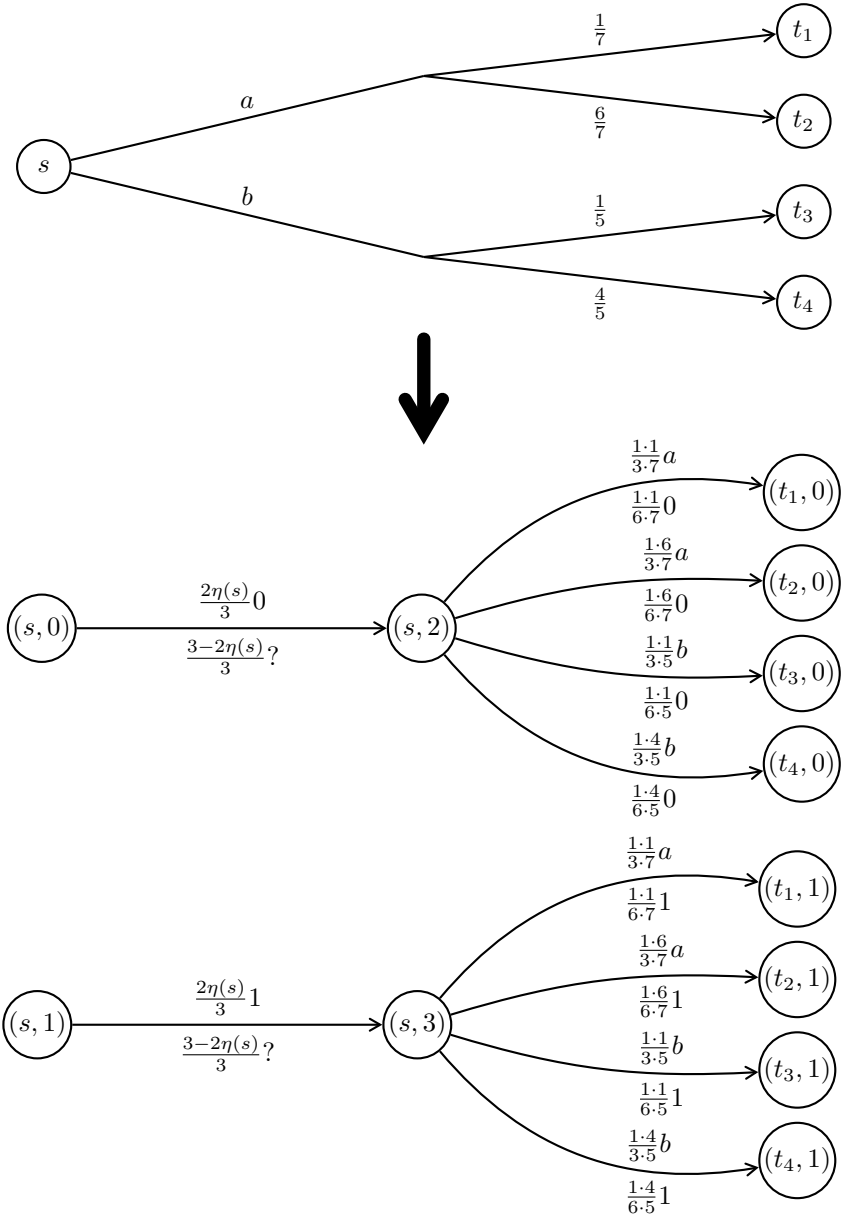
For all $s, t \in S_\P$ and both $i \in \{0,1\}$ we have the following transitions:
\begin{align*}
M(a)((s,i+2), (t,i)) \ &\defeq \ \frac23 \cdot \frac{1}{|\Sigma_\P|} \cdot M_\P(a)(s,t)  && \forall\,a \in \Sigma_\P \\
M(i)((s,i+2), (t,i))      \ &\defeq \ \frac13 \cdot \frac{1}{|\Sigma_\P|} \cdot \sum_{a \in \Sigma_\P} M_\P(a)(s,t)
\end{align*}
Note that any state $(s,i+2)$ emits, with probability~$\frac13$, the letter~$i$ (thus giving away in which copy $\M$~is), or, with probability~$\frac23$, a letter in~$\Sigma_\P$ that is chosen uniformly at random. (Thus, if $|\Sigma_\P|=2$, all letters in $\Sigma_\P \cup \{i\}$ have probability~$\frac13$.)
In more detail, one can view the behaviour of~$\M$ in a state $(s,i+2)$ as follows: $\M$~first samples a letter~$a$ from~$\Sigma_\P$ uniformly at random, then samples a successor state from $S \times \{i\}$ (each $(t,i)$ with probability $M_\P(a)(s,t)$), and then, upon transitioning to~$(t,i)$, emits, with probability~$\frac23$, the sampled letter~$a$, or, with probability~$\frac13$, the letter~$i$.
This completes the description of the construction.

There exists a diagnoser:
Consider the policy~$\overline\rho$, which, after the initial~$0$, observes every $2$nd letter; each time this letter will be either from~$\Sigma_\P$ (with probability~$\frac23$) or from~$\{0,1\}$ (with probability~$\frac13$).
With probability~$1$ the latter case occurs eventually.
Once it has occurred, the copy has been revealed; thus the observation prefix becomes deciding and $\overline\rho$~stops making any observations.
The required number of observations under~$\overline\rho$ is geometrically distributed, with expectation $1/\frac13 = 3$.
Hence $\Ex(C_{\overline\rho}) = 3$ and thus $\cinf \le 3$.

It remains to show that $\cinf < 3$ if and only if there is~$u$ such that $\Pr_{\P}(u) > \frac12$.
Suppose $u_0$~is such that $\Pr_{\P}(u_0) > \frac12$.
Consider the following policy~$\rho_0$:
\begin{itemize}
\item do not observe the initial~$0$ and then observe every $2$nd letter until either the prefix is deciding (because $0$ or~$1$ has been observed) or the observation prefix becomes equal to~$u_0$;
\item in the former case: stop observing forever;
\item in the latter case: observe the immediately following letter;
\begin{itemize}
 \item if this letter is $0$ or~$1$ then the observation prefix has become deciding, so stop observing;
 \item if this letter is~$?$ then observe the next letter and then every $2$nd letter as before.
\end{itemize}
\end{itemize}
This policy coincides mostly with the diagnoser~$\overline\rho$ described above.
It only deviates if and when $u_0$~has been observed.
One can show that this deviation improves the expected cost: the conditional probability that $0$ or~$1$ (hence not the letter~$?$) will be observed conditioned under having observed~$u_0$ is equal to $\frac23 \cdot \Pr_{\P}(u_0) > \frac23 \cdot \frac12 > \frac13$.
It follows that $\cinf < 3$.

It remains to show the converse.
To this end, suppose all~$u$ satisfy $\Pr_{\P}(u) \le \frac12$.
Then one is tempted to think that it is not beneficial to deviate from the diagnoser~$\overline\rho$ described above:
conditioned under having observed some observation prefix~$u$ using~$\overline\rho$, the conditional probability that the immediately following letter is $0$ or~$1$ (hence not the letter~$?$) is equal to $\frac23 \cdot \Pr_{\P}(u) \le \frac23 \cdot \frac12 \le \frac13$, so not better than what one would get by proceeding with $\overline\rho$.

The last argument is flawed though; it is valid only for the very first deviation from~$\overline{\rho}$.
Indeed, suppose there is some $u_0$ with $\Pr_{\P}(u_0) = \frac12$.
Suppose further we follow~$\overline{\rho}$ until the observation prefix is $u_0$.
At this point it does not come with a risk to observe the immediately following letter, as the probability to observe $0$ or~$1$ is~$\frac13$, which is equal to the success probability of each single observation of~$\overline{\rho}$.
Suppose we do that and observe~$?$.
Now we have not learned anything about which copy the MC is in; however, after this observation, the conditional probability (conditioned on all observations so far) that we are in a state $(s,2)$ or $(s,3)$ with \emph{low} $\eta(s)$ has increased, as a state $(s,i)$ with high~$\eta(s)$ would probably have produced the letter~$i$ instead of~$?$.
This information might be exploitable later:
For instance, suppose that the PA is such that states~$s$ with low and high values of $\eta(s)$ alternate.
Then in the next round we believe to be in a state $(s,i)$ with high $\eta(s)$.
Now it might be beneficial to make another non-$\overline\rho$ observation, even if $\Pr_{\P}(u_0 a) \le \frac12$ holds for all $a \in \Sigma_\P$.

To fix this problem we need to change the construction in a way that a single deviation from~$\overline\rho$ leaks less information.
We do this in the following \autoref{sub-correct-proof-undecidable}.
\end{proof}

\subsection{Correct Proof of \autoref{thm-undecidable}}
\label{sub-correct-proof-undecidable}

Here is \autoref{thm-undecidable} from the main body:
\thmundecidable*
\begin{proof}
We reduce from the emptiness problem for probabilistic automata.
A \emph{probabilistic automaton (PA)} is a tuple $\P = (S, \Sigma, M, s_0, \eta)$ where
 $S$ is a finite set of states,
 $\Sigma$ is a finite alphabet,
 the mapping $M: \Sigma \to [0,1]^{S \times S}$, where $M(a)$ is stochastic for each $a \in \Sigma$, specifies the transitions,
 $s_0$~is an initial state,
 and $\eta \in [0,1]^S$ is a vector of acceptance probabilities.
Extend~$M$ to $M: \Sigma^* \to [0,1]^{S \times S}$ as in the case of MCs.
In the case of PAs, $M(u)$ is stochastic for each $u \in \Sigma^*$.
For each~$u$ define $\Pr_{\P}(u) := e_{s_0} M(u) \eta^\top$.
The probability~$\Pr_{\P}(u)$ can be interpreted as the probability that $\P$ accepts~$u$.
The \emph{emptiness problem} asks, given a PA~$\P$, whether there is~$u$ such that $\Pr_{\P}(u) > \frac12$.
This problem is undecidable~\cite[p.~190, Theorem~6.17]{Paz71}.

We assume that for all~$u$ there is $s$ with $M(u)(s_0,s) > 0$ and $0 < \eta(s) < 1$.
This is without loss of generality, as we can make the PA branch, in its first transition and with positive probability, to a sub-PA with a single state from which every word is accepted with probability~$\frac12$.
Formally, from the original PA $\P = (S, \Sigma, M, s_0, \eta)$ obtain another PA $\P' = (S' \cup \{t_0, t\}, \Sigma, M', t_0, \eta')$ where:
\begin{itemize}
\item $t_0, t$ are fresh states;
\item $\eta'(t_0) \defeq \eta(s_0)$ and $\eta'(t) \defeq \frac12$ and $\eta'(s) \defeq \eta(s)$ for all $s$;
\item $M'(a)(t_0, t) \defeq \frac12$ and $M'(a)(t,t) \defeq 1$ for all $a \in \Sigma$;
\item $M'(a)(t_0, s) \defeq \frac12 M(a)(s_0,s)$ for all $a \in \Sigma$ and all $s$;
\item $M'(a)(s,s') = M(a)(s,s')$ for all $s,s'$.
\end{itemize}
We have for all~$u$ that $\Pr_{\P'}(u) - \frac12 = \frac12 (\Pr_{\P}(u) - \frac12)$.
It follows that
$\Pr_{\P}(u) > \frac12$ if and only if $\Pr_{\P'}(u) > \frac12$. 

Let $\P = (S_\P, \Sigma_\P, M_\P, s_{0\P}, \eta)$ be the given PA.
We will construct a DFA~$\A$ and an MC~$\M$ over the alphabet $\Sigma \defeq \Sigma_\P \cup \{0,1\}$, where $0,1$ are fresh letters.
Define $\A$ to be a DFA that accepts $L = 0 (\Sigma \Sigma \Sigma \Sigma \Sigma)^* \Sigma \Sigma \Sigma \Sigma 1 \Sigma^\omega$.
Ignoring the very first letter~$0$, we might characterize~$L$ by saying that $1$ appears on a position that is divisible by~$5$.

We construct an MC $\M = (S, \Sigma, M, s_0)$ with
\[
 S \defeq \{s_0\} \cup (S_\P \times \{0,1,2,3\}) \cup \tilde{S}\,,
\] 
such that $s_0$~and the states in~$\tilde{S}$ are fresh states.
%
The MC~$\M$ initially splits randomly into a ``$0$-copy'' (with states in $S \times \{0,2\}$) and a ``$1$-copy'' (with states in $S \times \{1,3\}$), in either case emitting the letter~$0$.
Formally, $M(0)(s_0, (s_{0\P},0)) = M(0)(s_0, (s_{0\P},1)) = \frac12$.
The MC~$\M$ is constructed such that if it goes into the $0$-copy then surely it emits an infinite word that is not in~$L$;
and if it goes into the $1$-copy then almost surely it emits an infinite word in~$L$.
Therefore, for an observation policy it suffices to identify which copy $\M$~has entered.

The transitions in~$\M$ depend on $M_\P$ and~$\eta$.
In \autoref{fig-illustration-undecidable} we illustrate this dependence with an example, where $\Sigma_\P = \{a,b\}$.
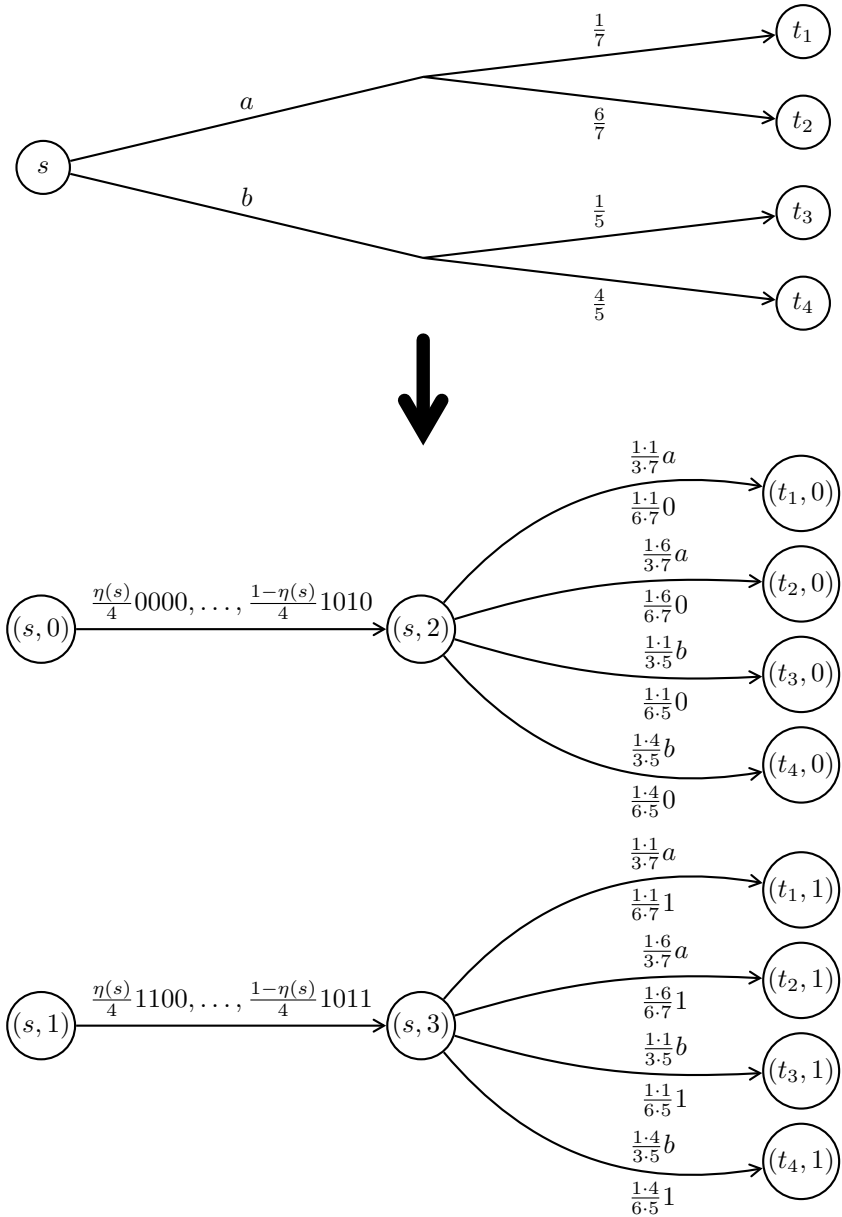
\begin{figure}
\begin{center}
\begin{tikzpicture}[scale=2,LMC style]
\node[state] (s)  at (0.5,0) {$s$};
\coordinate  (sa) at (3.0,0.6);
\coordinate  (sb) at (3.0,-0.6);
\node[state] (t1) at (5.5,0.9) {$t_1$};
\node[state] (t2) at (5.5,0.3) {$t_2$};
\node[state] (t3) at (5.5,-0.3) {$t_3$};
\node[state] (t4) at (5.5,-0.9) {$t_4$};
\path (s) edge node[above] {$a$} (sa);
\path (s) edge node[above] {$b$} (sb);
\path[->] (sa) edge node[above] {$\frac17$} (t1);
\path[->] (sa) edge node[below] {$\frac67$} (t2);
\path[->] (sb) edge node[above] {$\frac15$} (t3);
\path[->] (sb) edge node[below] {$\frac45$} (t4);
\end{tikzpicture} \\
\begin{tikzpicture}[scale=2,LMC style]
\coordinate (0) at (0,0);
\path[->] (0) edge[line width=5, line cap=round] (0,-0.7);
\end{tikzpicture}
\\[-3mm]
\begin{tikzpicture}[scale=2,LMC style]
\node[state] (s0)  at (0.5,0) {$(s,0)$};
\node[state] (s2)  at (3,0) {$(s,2)$};
\path[->] (s0) edge node[above] {$\frac{\eta(s)}{4} 0000, \ldots, \frac{1-\eta(s)}{4} 1010$} (s2);
\node[state] (t10) at (5.5,0.9) {$(t_1, 0)$};
\node[state] (t20) at (5.5,0.3) {$(t_2, 0)$};
\node[state] (t30) at (5.5,-0.3) {$(t_3, 0)$};
\node[state] (t40) at (5.5,-0.9) {$(t_4, 0)$};
\path[->] (s2) edge[bend left=30]  node[pos=0.7,above] {$\frac{1 \cdot 1}{3 \cdot 7} a$} node[pos=0.7,below] {$\frac{1 \cdot 1}{6 \cdot 7} 0$} (t10);
\path[->] (s2) edge[bend left=10]  node[pos=0.7,above] {$\frac{1 \cdot 6}{3 \cdot 7} a$} node[pos=0.7,below] {$\frac{1 \cdot 6}{6 \cdot 7} 0$} (t20);
\path[->] (s2) edge[bend right=10] node[pos=0.7,above] {$\frac{1 \cdot 1}{3 \cdot 5} b$} node[pos=0.7,below] {$\frac{1 \cdot 1}{6 \cdot 5} 0$} (t30);
\path[->] (s2) edge[bend right=30] node[pos=0.7,above] {$\frac{1 \cdot 4}{3 \cdot 5} b$} node[pos=0.7,below] {$\frac{1 \cdot 4}{6 \cdot 5} 0$} (t40);
\end{tikzpicture}
\\
\begin{tikzpicture}[scale=2,LMC style]
\node[state] (s0)  at (0.5,0) {$(s,1)$};
\node[state] (s2)  at (3,0) {$(s,3)$};
\path[->] (s0) edge node[above] {$\frac{\eta(s)}{4} 1100, \ldots, \frac{1-\eta(s)}{4} 1011$} (s2);
\node[state] (t10) at (5.5,0.9) {$(t_1, 1)$};
\node[state] (t20) at (5.5,0.3) {$(t_2, 1)$};
\node[state] (t30) at (5.5,-0.3) {$(t_3, 1)$};
\node[state] (t40) at (5.5,-0.9) {$(t_4, 1)$};
\path[->] (s2) edge[bend left=30]  node[pos=0.7,above] {$\frac{1 \cdot 1}{3 \cdot 7} a$} node[pos=0.7,below] {$\frac{1 \cdot 1}{6 \cdot 7} 1$} (t10);
\path[->] (s2) edge[bend left=10]  node[pos=0.7,above] {$\frac{1 \cdot 6}{3 \cdot 7} a$} node[pos=0.7,below] {$\frac{1 \cdot 6}{6 \cdot 7} 1$} (t20);
\path[->] (s2) edge[bend right=10] node[pos=0.7,above] {$\frac{1 \cdot 1}{3 \cdot 5} b$} node[pos=0.7,below] {$\frac{1 \cdot 1}{6 \cdot 5} 1$} (t30);
\path[->] (s2) edge[bend right=30] node[pos=0.7,above] {$\frac{1 \cdot 4}{3 \cdot 5} b$} node[pos=0.7,below] {$\frac{1 \cdot 4}{6 \cdot 5} 1$} (t40);
\end{tikzpicture}
\end{center}
\caption{Illustration of the reduction. The PA at the top has a state~$s$ whose outgoing transitions are such that $M(a)(s,t_1) = \frac17$, $M(a)(s,t_2) = \frac67$, $M(b)(s,t_3) = \frac15$, $M(b)(s,t_4) = \frac45$.
The resulting MC has two corresponding states, $(s,0)$ and $(s,1)$.
From such a state $(s,i)$ it first emits, with a probability depending on $\eta(s)$, a $4$-bit sequence (``block''), and then emits, with probability~$\frac13$, the letter~$i$ (thus giving away in which copy the MC is), or, with probability $\frac23 \cdot \frac{1}{|\Sigma_\P|} = \frac13$ each, the letters $a$ or~$b$.
}
\label{fig-illustration-undecidable}
\end{figure}

Formally, for all $(s,i) \in S_\P \times \{0,1\}$ we define transitions (using states in~$\tilde{S}$ that are not further specified) such that:
\begin{equation}
\begin{aligned}
M(0000)((s,0),(s,2)) \, &= \, \eta(s)/4  &\ M(0101)((s,0),(s,2)) \, &= \, (1-\eta(s))/4 \\
M(0001)((s,0),(s,2)) \, &= \, \eta(s)/4  &\ M(0110)((s,0),(s,2)) \, &= \, (1-\eta(s))/4 \\
M(0010)((s,0),(s,2)) \, &= \, \eta(s)/4  &\ M(1001)((s,0),(s,2)) \, &= \, (1-\eta(s))/4 \\
M(0011)((s,0),(s,2)) \, &= \, \eta(s)/4  &\ M(1010)((s,0),(s,2)) \, &= \, (1-\eta(s))/4 \\[3mm]
M(1100)((s,1),(s,3)) \, &= \, \eta(s)/4  &\ M(0100)((s,1),(s,3)) \, &= \, (1-\eta(s))/4 \\
M(1101)((s,1),(s,3)) \, &= \, \eta(s)/4  &\ M(0111)((s,1),(s,3)) \, &= \, (1-\eta(s))/4 \\
M(1110)((s,1),(s,3)) \, &= \, \eta(s)/4  &\ M(1000)((s,1),(s,3)) \, &= \, (1-\eta(s))/4 \\
M(1111)((s,1),(s,3)) \, &= \, \eta(s)/4  &\ M(1011)((s,1),(s,3)) \, &= \, (1-\eta(s))/4
\end{aligned}
\label{eq-def-blocks}
\end{equation}
These $4$-bit sequences (henceforth ``blocks'') are chosen such that if the first two bits agree (say they are both $b \in \{0,1\}$) then $\M$ is in the $b$-copy;
if the first two bits of a block do not agree then the xor of all four bits (equivalently, the parity of the sum of the four bits) identifies in which copy $\M$~is.
States~$(s,i)$ with higher acceptance probability~$\eta(s)$ have a higher chance to emit a block where the first two bits agree.
The intention of this construction is to make it cost-efficient for an observation policy to observe (some) letters of a block when the policy believes that it is in a state $(s,i)$ with high~$\eta(s)$.
We will show later that this is beneficial on some runs if and only if there is~$u$ with $\Pr_{\P}(u) > \frac12$.

The other transitions of~$\M$ are as follows:
For all $s, t \in S_\P$ and both $i \in \{0,1\}$ we have the following transitions:
\begin{align*}
M(a)((s,i+2), (t,i)) \ &\defeq \ \frac23 \cdot \frac{1}{|\Sigma_\P|} \cdot M_\P(a)(s,t)  && \forall\,a \in \Sigma_\P \\
M(i)((s,i+2), (t,i))      \ &\defeq \ \frac13 \cdot \frac{1}{|\Sigma_\P|} \cdot \sum_{a \in \Sigma_\P} M_\P(a)(s,t)
\end{align*}
Note that any state $(s,i+2)$ emits, with probability~$\frac13$, the letter~$i$ (thus giving away in which copy $\M$~is), or, with probability~$\frac23$, a letter in~$\Sigma_\P$ that is chosen uniformly at random. (Thus, if $|\Sigma_\P|=2$, all letters in $\Sigma_\P \cup \{i\}$ have probability~$\frac13$.)
In more detail, one can view the behaviour of~$\M$ in a state $(s,i+2)$ as follows: $\M$ first samples a letter~$a$ from~$\Sigma_\P$ uniformly at random, then samples a successor state from $S \times \{i\}$ (each $(t,i)$ with probability $M_\P(a)(s,t)$), and then, upon transitioning to~$(t,i)$, emits, with probability~$\frac23$, the sampled letter~$a$, or, with probability~$\frac13$, the letter~$i$.
This completes the description of the construction.

There exists a diagnoser:
Consider the policy~$\overline\rho$, which, after the initial~$0$, observes every $5$th letter; each time this letter will be either from~$\Sigma_\P$ (with probability~$\frac23$) or from~$\{0,1\}$ (with probability~$\frac13$).
With probability~$1$ the latter case occurs eventually.
Once it has occurred, the copy has been revealed; thus the observation prefix becomes deciding and $\overline\rho$~stops making any observations.
The required number of observations under~$\overline\rho$ is geometrically distributed, with expectation $1/\frac13 = 3$.
Hence $\Ex(C_{\overline\rho}) = 3$ and thus $\cinf \le 3$.

It remains to show that $\cinf < 3$ if and only if there is~$u$ such that $\Pr_{\P}(u) > \frac12$.
Suppose $u_0$~is such that $\Pr_{\P}(u_0) > \frac12$.
Consider the following policy~$\rho_0$:
\begin{itemize}
\item do not observe the initial~$0$ and then observe every $5$th letter until either the prefix is deciding (because $0$ or~$1$ has been observed) or the observation prefix becomes equal to~$u_0$;
\item in the former case: stop observing forever;
\item in the latter case: observe the first two letters of the following block;
\begin{itemize}
 \item if they agree 
       then the observation prefix has become deciding, so stop observing;
 \item if they do not agree, observe also the next two letters; this produces a deciding prefix, so stop observing.
\end{itemize}
\end{itemize}
This policy coincides mostly with the diagnoser~$\overline\rho$ described above.
It only deviates if and when $u_0$~has been observed.
We show that this deviation improves the expected cost.

To argue in more formal terms,
we extend (for the PA~$\P$) the mapping $M_\P: \Sigma_\P^* \to [0,1]^{S_\P \times S_\P}$ to $M_\P: (\Sigma_\P \cup \{\bot\})^* \to [0,1]^{S_\P \times S_\P}$  by defining
\begin{align*}
 M_\P(\upsilon)   \ &\defeq \ {\displaystyle\sum_{u \sim \upsilon} M_\P(u)} \Big/ {\left|\{u \mid u \sim \upsilon\}\right|} && \text{and define} \\
 \mu_\P(\upsilon) \ &\defeq \ e_{s_{0\P}} M_\P(\upsilon) \;.
\intertext{%
One can view the distribution~$\mu_\P(\upsilon)$ as the expected distribution after having fed the PA~$\P$ with a randomly sampled $u \sim \upsilon$.
Similarly, we extend (for the MC~$\M$) the mapping $M: \Sigma^* \to [0,1]^{S \times S}$ to $M: \Sigma_\bot^* \to [0,1]^{S \times S}$  by defining
}
 M(\upsilon)   \ &\defeq \ \sum_{u \sim \upsilon} M(u) && \text{and define} \\
 \mu(\upsilon) \ &\defeq \  \frac{e_{s_0} M(\upsilon)}{e_{s_0} M(\upsilon) \vec{1}^\T} \;.
\end{align*}
One can view the distribution~$\mu(\upsilon)$ as the expected distribution after having observed~$\upsilon$.
Finally, for any finite word~$\upsilon = o_1 o_2 \cdots o_n \in (\Sigma_\P \cup \{\bot\})^*$ define the following padding:
\[
 \widehat{\upsilon} \ \defeq \ \begin{cases} \bot & \text{if $n=0$} \\
                                         \bot \bot \bot \bot \bot o_1 \bot \bot \bot \bot o_2 \bot \bot \bot \bot o_3 \cdots \bot \bot \bot \bot o_n
                                           & \text{if $n \ge 1$}
                           \end{cases}
\]
With these definitions, it is straightforward to check that we have:
\begin{equation} \label{eq-undec-conditioning}
 \frac12 \mu_\P(\upsilon)(s) \ = \ \mu(\widehat{\upsilon})((s,0)) \ = \ \mu(\widehat{\upsilon})((s,1)) \qquad
 \forall\,\upsilon \in (\Sigma_\P \cup \{\bot\})^* \ \forall\,s \in S_\P
\end{equation}

For the word~$u_0$ from above we have $\mu_\P(u_0) \eta^\T = \Pr_\P(u_0) > \frac12$.
It follows from \eqref{eq-undec-conditioning} and~\eqref{eq-def-blocks} that, conditioned under prefix~$\widehat{u_0}$, the conditional probability of emitting $00$ or~$11$ at the beginning of the following block is greater than~$\frac12$; formally:
\[
\Pr(\{u 0 0, u 1 1\} \Sigma^\omega \mid u \sim \widehat{u_0} \} \mid \{w \gtrsim \widehat{u_0}\}) \ = \ \mu_\P(u_0) \eta^\T \ > \ \frac12
\]
Recall that $\rho_0$ is defined so that once it has observed~$\widehat{u_0}$, it makes either exactly~$2$ or exactly~$4$ further observations.
From the previous inequality it follows for~$\rho_0$ that, conditioned under observing~$\widehat{u_0}$, the conditional probability to make exactly~$2$ further observations is greater than the conditional probability to make exactly~$4$ further observations.
Hence, the conditional expected number of observations after having observed~$u_0$ is less than~$3$.

More formally, for any policy~$\rho$ and an observation prefix~$\upsilon = o_1 \cdots o_k$ define a random variable~$C_\rho^\upsilon$ by
\[
 C_\rho^\upsilon(w) \ \defeq \ C_\rho(w) - \sum_{k=0}^{|\upsilon|-1} \rho(o_1\ldots o_k) \qquad \text{for all } w \gtrsim \upsilon\,,
\]
i.e., $C_\rho^\upsilon$ is the number of observations that $\rho$~makes after~$\upsilon$.
By the argument above we have $\Ex(C_{\rho_0}^{\widehat{u_0}} \mid \{w \gtrsim \widehat{u_0}\}) < 3$.
For all~$u$ that are not prefixes of~$u_0$ we have $\Ex(C_{\rho_0}^{\widehat{u}} \mid \{w \gtrsim \widehat{u}\}) = 3$.
It follows $\Ex(C_{\rho_0}) < 3$.
Hence $\cinf < 3$.

It remains to show the converse.
To this end, suppose all~$u$ satisfy $\Pr_{\P}(u) \le \frac12$.
It suffices to show that all feasible policies~$\rho$ satisfy $\Ex(C_{\rho}) \ge 3$.

In the following we use regular expressions to describe observation prefixes.
For improved readability we may indicate the borders of a block with a dot.
For instance, the regular expression
\[
 \bot (. \bot^4 . (\Sigma_\P + \bot))^* . 0 \bot \bot \bot .
\]
indicates observation prefixes of the following form: first all observations (if any) are made in non-blocks and are in~$\Sigma_\P$ (i.e., are not $0$ or~$1$), and then $0$~is observed at the beginning of a block, and the other three letters in the block are not observed.

We call $k \in (0 + 1 + \bot)^4$ \emph{block-deciding} when $k \in (0+1)^4 + (00 + 11) (0+1+\bot)^2$.
We have the following lemma:
\begin{ourlemma} \label{lem-block-deciding}
Let $\upsilon$ be an observation prefix that is not deciding and satisfies $|\upsilon| = 1 + 5 n$ for some $n \in \NN$.
Let $k \in (0 + 1 + \bot)^4$.
Then $\upsilon . k .$ is deciding if and only if $k$ is block-deciding.
\end{ourlemma}
\begin{proof}[Proof of the lemma]
Let $\upsilon$ and~$k$ be as in the statement.
Suppose $k$~is block-deciding.
Then, by inspecting~\eqref{eq-def-blocks}, it follows that $\upsilon .k.$ is deciding.

Conversely, suppose $k$ is not block-deciding.
From the assumption made about~$\P$ in the beginning we obtain that there is $s \in S_\P$ with $0 <\eta(s) < 1$ such that $M(\upsilon)(s_0,(s,0)) > 0$ and $M(\upsilon)(s_0,(s,1)) > 0$.
By inspecting~\eqref{eq-def-blocks}, it follows that $\upsilon .k.$ is not deciding.
\end{proof}

Let $\rho$~be any feasible policy, and let $\upsilon$~be an observation prefix that $\rho$~produces with positive probability.
\begin{enumerate}
\item \label{enum-long} Suppose $\upsilon \in \bot (. \bot^4 . (\Sigma_\P + \bot))^* . k . (\bot^5)^* (\Sigma_\P+\bot) (\bot^5)^* .\ell$ where $k \in (0+1+\bot)^4$ is not block-deciding, and $\ell \in (0+1+\bot)^{\le 4}$ has exactly one observation.
    By \autoref{lem-block-deciding} the prefix up to and including~$k$ is not deciding.
    Another application of \autoref{lem-block-deciding} shows that $\upsilon$~is not deciding, so $\Ex(C_\rho^\upsilon) \ge 1$.
\item \label{enum-block-a} Suppose $\upsilon \in \bot (. \bot^4 . (\Sigma_\P + \bot))^* . k . (\bot^5)^* \Sigma_\P$ where $k \in (0+1+\bot)^4$ is not block-deciding.
    Similarly as before, it follows that $\upsilon$~is not deciding.
    \begin{itemize}
    \item Suppose the next observation that $\rho$~makes is in a block.
    Then, by \autoref{enum-long}, we have $\Ex(C_\rho^\upsilon) \ge 1+1 = 2$.
    \item Otherwise the next observation is in a non-block.
    If this observation yields $0$ or~$1$ (which happens with probability~$\frac13$) then no further observation is needed; otherwise, a further observation will be needed.
    So $\Ex(C_\rho^\upsilon) \ge 1 + \frac13 \cdot 0 + \frac23 \cdot 1 \ge \frac53$.
    \end{itemize}
    Thus, in either case we have $\Ex(C_\rho^\upsilon) \ge \frac53$.
\item \label{enum-block-4} Suppose $\upsilon \in \bot (. \bot^4 . (\Sigma_\P + \bot))^* . k .$ where $k \in (0+1+\bot)^4$ is not block-deciding.
    Similarly as before, it follows that $\upsilon$~is not deciding.
    \begin{itemize}
    \item Suppose the next observation that $\rho$~makes is in a block.
    Then, by \autoref{enum-long}, we have $\Ex(C_\rho^\upsilon) \ge 1+1 = 2$.
    \item Otherwise the next observation is in a non-block.
    If this observation yields $0$ or~$1$ (which happens with probability~$\frac13$) then no further observation is needed; otherwise, the resulting observation prefix will have the form from \autoref{enum-block-a}.
    It follows from \autoref{enum-block-a} that $\Ex(C_\rho^\upsilon) \ge 1 + \frac13 \cdot 0 + \frac23 \cdot \frac53 > 2$.
    \end{itemize}
    Thus, in either case we have $\Ex(C_\rho^\upsilon) \ge 2$.
\item \label{enum-block-3} Suppose $\upsilon \in \bot (. \bot^4 . (\Sigma_\P + \bot))^* . \bot \bot b$ for $b \in \{0,1\}$.
    Then $\upsilon$~is not deciding.
    \begin{itemize}
    \item Suppose the next observation is not in the same block.
    Then the resulting observation prefix has the form from either \autoref{enum-long} or \autoref{enum-block-a}.
    It follows from these items that $\Ex(C_\rho^\upsilon) \ge 1 + \min\{1, \frac53\} = 2$.
    \item Otherwise the next observation follows immediately and yields $b' \in \{0,1\}$.
    The word $\bot \bot b b'$ is not block-deciding, hence, by \autoref{lem-block-deciding}, $\upsilon b'$ is not deciding, so $\Ex(C_\rho^\upsilon) \ge 2$.
    \end{itemize}
    Thus, in either case we have $\Ex(C_\rho^\upsilon) \ge 2$.
\item \label{enum-block-2} Suppose $\upsilon \in \bot (. \bot^4 . (\Sigma_\P + \bot))^* . b b'$ for $b b' \in (0+1+\bot)^2$ and $b b' \not\in 00+11$.
    Then $\upsilon$~is not deciding.
    \begin{itemize}
    \item Suppose the next observation is not in the same block.
    Then the resulting observation prefix has the form from either \autoref{enum-long} or \autoref{enum-block-a}.
    It follows from these items that $\Ex(C_\rho^\upsilon) \ge 1 + \min\{1, \frac53\} = 2$.
    \item Otherwise the next observation results in an observation prefix $\upsilon \upsilon'$ with $\upsilon' \in \{0, 1, \bot 0, \bot 1\}$.
    No word in $\{b b'\} \{0 \bot, 1 \bot, \bot 0, \bot 1\}$ is block-deciding.
    Hence, by \autoref{lem-block-deciding}, $\upsilon \upsilon'$ is not deciding, so $\Ex(C_\rho^\upsilon) \ge 2$.
    \end{itemize}
    Thus, in either case we have $\Ex(C_\rho^\upsilon) \ge 2$.
\item \label{enum-block-1} Suppose $\upsilon \in \bot (. \bot^4 . (\Sigma_\P + \bot))^* . b$ for $b \in \{0,1\}$.
    Then $\upsilon$~is not deciding.
    \begin{itemize}
    \item Suppose the next observation is not in the same block.
    Then the resulting observation prefix has the form from either \autoref{enum-long} or \autoref{enum-block-a}.
    It follows from these items that $\Ex(C_\rho^\upsilon) \ge 1 + \min\{1, \frac53\} = 2$.
    \item Suppose the next observation is in the same block but does not follow immediately.
    This results in an observation prefix $\upsilon \bot \upsilon'$ with $\upsilon' \in \{0, 1, \bot 0, \bot 1\}$.
    No word in $\{b \bot\} \{0 \bot, 1 \bot, \bot 0, \bot 1\}$ is block-deciding.
    Hence, by \autoref{lem-block-deciding}, $\upsilon \upsilon'$ is not deciding, so $\Ex(C_\rho^\upsilon) \ge 2$.
    \item Otherwise the next observation follows immediately. 
    Let $\upsilon_0$~be such that $\upsilon = \upsilon_0 b$.
    Then there exists $\upsilon_1 \in (\Sigma_\P \cup \{\bot\})^*$ with $\upsilon_0 = \widehat{\upsilon_1}$.
    Since all $u \in \Sigma_\P^*$ satisfy $\mu_\P(u) \eta^\T = \Pr_{\P}(u) \le \frac12$, we also have $\mu_\P(\upsilon_1) \eta^\T \le \frac12$.
    It follows from \eqref{eq-undec-conditioning} and~\eqref{eq-def-blocks} that, conditioned under prefix~$\upsilon_0$, the conditional probability of emitting $0 0$ or $1 1$ is at most~$\frac12$; formally:
    \[
    p \ \defeq \ \Pr(\{u 0 0, u 1 1\} \Sigma^\omega \mid u \sim \upsilon_0 \} \mid \{w \gtrsim \upsilon_0\}) \ = \ \mu_\P(\upsilon_1) \eta^\T \ \le \ \frac12
    \]
    For symmetry reasons we must have:
    \[
      \Pr(\{u b\} \Sigma^\omega \mid u \sim \upsilon \} \mid \{w \gtrsim \upsilon\})
      \ = \ p \ \le \ \frac12
    \]
    In words, the conditional probability that the next observation equals~$b$ is at most~$\frac12$.
    If this happens then no further observation is needed, as $b b \bot \bot$ is block-deciding;
    otherwise, the resulting observation prefix has the form from \autoref{enum-block-2}.
    It follows from \autoref{enum-block-2} that $\Ex(C_\rho^\upsilon) \ge 1 + p \cdot 0 + (1-p) \cdot 2 \ge 1 + \frac12 \cdot 2 = 2$.
    \end{itemize}
    Thus, in all three cases we have $\Ex(C_\rho^\upsilon) \ge 2$.
\end{enumerate}
Let $\rho$~be feasible.
Call an observation prefix~$\upsilon$ \df{conventional} when
\begin{itemize}
\item $\upsilon$ is not deciding;
\item $\rho$ produces~$\upsilon$ with positive probability;
\item and $\upsilon$ does not contain observations in blocks.
\end{itemize}
Towards a contradiction, assume that $\Ex(C_\rho) < 3$.
Then there exist a conventional observation prefix~$\upsilon$ and $x>0$ such that:
\begin{itemize}
\item $\Ex(C_\rho^\upsilon) \le 3 - x$;
\item and for all conventional observation prefixes~$\upsilon'$ we have $\Ex(C_\rho^{\upsilon'}) > 3 - \frac32 x$.
\end{itemize}
Since $\rho$~is feasible, there must be a non-$\bot$ observation (at some point) after~$\upsilon$.
\begin{itemize}
\item Suppose this next observation is in a non-block, resulting in an observation prefix~$\upsilon'$.
Then $\upsilon'$ is either deciding (with probability~$\frac13$) or conventional (with probability~$\frac23$).
It follows:
\[
 3 - x \ \ge \ \Ex(C_\rho^\upsilon) \ \ge \ 1 + \frac13 \cdot 0 + \frac23 \cdot \Ex(C_\rho^{\upsilon'}) \ > \ 1 + \frac23 \cdot \left(3 - \frac32 x\right) \ = \ 3 - x\,,
\]
which is a contradiction.
\item Otherwise, the next observation is in a block, resulting in an observation prefix~$\upsilon'$.
Then $\upsilon'$ is not deciding and has the form of items \ref{enum-block-4}--\ref{enum-block-1} above.
It follows from these items that $\Ex(C_\rho^\upsilon) = 1 + \Ex(C_\rho^{\upsilon'}) \ge 1 + 2 = 3$, contradicting that $\Ex(C_\rho^\upsilon) \le 3 - x$.
\end{itemize}
In either case we have a contradiction.
Hence $\Ex(C_\rho) \ge 3$.
\end{proof}

\subsection{Proof of \autoref{lem-confused-not-bisimilar}}

Here is \autoref{lem-confused-not-bisimilar} from the main body:
\lemconfusednotbisimilar*
\begin{proof}
Since singleton beliefs cannot be confused in a non-hidden MC, we have:
\begin{equation} \label{eq-ergodic}
 \Pr_{\t{a}}(\{u w \mid \Pr_{\t{a u}}(L_{\delta(q,u)}) \in \{0,1\}\}) \ = \ 1 \quad \forall\,a \ \forall\,q 
\end{equation}

Suppose $\Delta(B,a)$ is settled for all~$a$.
Then for all~$a$ and all $(\t{a},q_1), (\t{a},q_2) \in \Delta(B,a)$ and all~$u$ with $\t{a} \en{u}$ we have:
\[
 \Pr_{\t{a u}}(L_{\delta(q_1,u)}) = 1 \quad \Longleftrightarrow \quad \Pr_{\t{a u}}(L_{\delta(q_2,u)}) = 1
\]
By combining this with~\eqref{eq-ergodic} we get:
\begin{align*}
 \Pr_{\t{a}}(\{u w \mid \Delta(B, a u) \text{ is deciding}\}) \ &= \ 1 \quad \forall\,a
\intertext{Hence:}
 \Pr_{s}(\{u w \mid \Delta(B, u) \text{ is deciding}\}) \ &= \ 1 \quad \forall\,(s,q) \in B
\end{align*}
Hence $B$~is not confused.

Conversely, suppose there is~$a$ such that $\Delta(B,a)$ is not settled.
Then there are $(\t{a},q_1), (\t{a},q_2) \in \Delta(B,a)$ and~$u_0$ with $\t{a} \en{u_0}$ such that:
\[
 \Pr_{\t{a u_0}}(L_{\delta(q_1,u_0)}) \ = \ 1 \ > \ \Pr_{\t{a u_0}}(L_{\delta(q_2,u_0)})
\]
It follows from~\eqref{eq-ergodic} that there is~$u_1$ with $\t{a u_0} \en{u_1}$ such that:
\[
 \Pr_{\t{a u_0 u_1}}(L_{\delta(q_1,u_0 u_1)}) \ = \ 1 \ > \ 0 \ = \ \Pr_{\t{a u_0 u_1}}(L_{\delta(q_2,u_0 u_1)})
\]
Hence for all~$u_2$ with $\t{a u_0 u_1} \en{u_2}$ we have that $(\t{a u_0 u_1 u_2},\delta(q_1,u_0 u_1 u_2))$ is positively deciding and $(\t{a u_0 u_1 u_2},\delta(q_2,u_0 u_1 u_2))$ is negatively deciding, and therefore also that
\[
 \Delta(B,a u_0 u_1 u_2) \supseteq \{(\t{a u_0 u_1 u_2},\delta(q_1,u_0 u_1 u_2)), (\t{a u_0 u_1 u_2},\delta(q_1,u_0 u_1 u_2))\}
\]
is not deciding.
So we have:
\[
 \Pr_{\t{a u_0 u_1}}(\{u_2 w \mid \Delta(B, a u_0 u_1 u_2) \text{ is deciding}\}) \ = \ 0
\]
It follows that there is $(s,q) \in B$ with 
\[
 \Pr_{s}(\{u w \mid \Delta(B, u) \text{ is deciding}\}) \ < \ 1
\]
Hence $B$~is confused.
\end{proof}

\subsection{Proof of \autoref{lem-last-obs}}

Here is \autoref{lem-last-obs} from the main body:
\lemlastobs*
\begin{proof}
Let $(s,q) \in S \times Q$.
The singleton belief $\{(s,q)\}$ cannot be confused in a non-hidden MC, i.e., we have:
\begin{align} 
1 \ 
&= \ \Pr_s(\{u a w \mid (\t{a},\delta(q,u a)) \text{ is deciding} \}) \notag 
\intertext{It follows that for each~$p<1$ there is~$K \in \NN$ such that}
p \ &\le \
\Pr_s(\{u a w \mid (\t{a},\delta(q,u a)) \text{ is deciding}, \ |u| \le K \}) \notag \\
&= \ \Pr_s(\{u a w \mid (\t{a},\delta(q,u a)) \text{ is deciding}, \ |u| = K \}) \notag
\intertext{%
Suppose $\cras(s,q) = \infty$.
Then for all $k \in \NN$ the belief $\Delta(\{(s,q)\}, \bot^k)$ is not confused, and so, by \autoref{lem-confused-not-bisimilar}, for all~$a$ the belief $\Delta(\{(s,q)\}, \bot^k a)$ is settled.
It follows that for each~$p<1$ there is~$K \in \NN$ such that
}
p \ &\le \
\Pr_s(\Sigma^K \{a w \mid \Delta(\{(s,q)\}, \bot^K a) \text{ is deciding} \}) \label{eq-last-obs}
\end{align}
Since there are finitely many $(s,q)$, it is also the case that for each $p<1$ there is $K \in \NN$ such that for all $(s,q)$ Equation~\eqref{eq-last-obs} holds.

Let $\varepsilon > 0$.
Choose $K \in \NN$ such that \eqref{eq-last-obs} holds for $p \defeq 1 / (1+\varepsilon)$.
Let $(s,q)$ be such that $\cras(s,q) = \infty$ and let $X$ denote the random number of observations that $\prho(K)$~needs to make starting in~$(s,q)$.
From the argument above, we have $\Pr_s(X > 1) \le 1-p$.
Since for the (random) pair $(s',q')$ after the next observation (i.e., $\{(s',q')\} = \Delta(\{(s,q)\}, \bot^K a)$ for the observed~$a$) we have $\cras(s',q') = \infty$ again, it follows that $\Pr_s(X > i) \le (1-p)^i$ holds for all $i \in \NN$.
Hence we have:
\[
c(K,s,q) \ \defeq \ \Ex_s(X) \ = \ \sum_{i=0}^\infty \Pr_s(X > i) \ \le \ \sum_{i=0}^\infty (1-p)^i \ = \ \frac1p \ = \ 1+\varepsilon
\]
\end{proof}

\subsection{Proof of \autoref{thm-max-pro-optimal}}

Here is \autoref{thm-max-pro-optimal} from the main body:
\thmmaxprooptimal*
\begin{proof}
Let $\rho$ be a feasible policy.
First we compare $\prho \defeq \prho(\infty)$ with~$\rho$.
Let $w$ be such that all its prefixes are enabled.
For $n \in \NN$, define $\up(n)(w)$ and $\ur(n)(w)$ as the observation prefixes obtained by $\prho$ and~$\rho$, respectively, after $n$~steps.
(Thus, $\up(n)(w)$ and $\ur(n)(w)$ are, respectively, the length $n$ prefixes of $\pi_{\prho}(w)$ and $\pi_{\rho}(w)$.)
We write $\lp(n)(w)$ and $\lr(n)(w)$ for the number of non-$\bot$ observations in $\up(n)(w)$ and $\ur(n)(w)$, respectively.
Define $\Bp(n)(w) \defeq \Delta(B_0,\up(n)(w))$ and $\Br(n)(w) \defeq \Delta(B_0,\ur(n)(w))$.
For beliefs $B, B'$ we write $B \preceq B'$ when for all $(s,q) \in B$ there is $(s',q') \in B'$ with $(s,q) \bis (s',q')$.
In the following we suppress~$w$ in the notation to avoid clutter.
We show for all~$w$ and all $n \in \NN$: 
\begin{equation} \label{eq-greedy-induction}
 \lp(n) \le \lr(n) \quad \text{ and } \quad \big(\Bp(n) \preceq \Br(n) \quad \text{or} \quad \lp(n) < \lr(n)\big)
\end{equation}
We proceed by induction on~$n$.
In the base case, $n=0$, we have $\lp(0) = 0 = \lr(0)$ and $\Bp(0) = B_0 = \Br(0)$.
For the inductive step, suppose \eqref{eq-greedy-induction} holds for~$n$.
(Note that, if $\upsilon_1\sim u$ and $\upsilon_2\sim u$ and $B_1 \preceq B_2$ and $|\res{\Delta(B_1,\upsilon_1)}|\le 1$, then $\Delta(B_1,\upsilon_1) \preceq \Delta(B_2,\upsilon_2)$.)
\begin{enumerate}
\item Suppose $\prho$ does not observe the $(n+1)$st letter.
Then $\lp(n+1) = \lp(n)$.
\begin{enumerate}
\item Suppose $\lp(n+1) < \lr(n+1)$. Then \eqref{eq-greedy-induction}~holds for $n+1$.
\item Otherwise, $\lr(n+1) \le \lp(n+1) = \lp(n) \le \lr(n) \le \lr(n+1)$ by the induction hypothesis.
So all these numbers are equal and we conclude that $\rho$~does not observe the $(n+1)$st letter either and that $\lp(n) = \lr(n)$.
It follows by the induction hypothesis that $\Bp(n) \preceq \Br(n)$.
Thus $\Bp(n+1) = \Delta(\Bp(n), \bot) \preceq \Delta(\Br(n), \bot) = \Br(n+1)$.
Hence \eqref{eq-greedy-induction}~holds for $n+1$.
\end{enumerate}
\item Otherwise $\prho$ observes the $(n+1)$st letter.
From the definition of~$\prho$ it follows that $\Bp(n+1)$ is settled and thus $\Bp(n+1) \preceq \Br(n+1)$.
\begin{enumerate}
\item Suppose $\lp(n) < \lr(n)$. Then $\lp(n+1) \le \lr(n) \le \lr(n+1)$, i.e., \eqref{eq-greedy-induction}~holds for $n+1$.
\item Otherwise, by the induction hypothesis, we have $\Bp(n) \preceq \Br(n)$.
Since $\prho$ observes the $(n+1)$st letter, the belief $\Delta(\Bp(n),\bot)$ is confused.
Since $\Bp(n) \preceq \Br(n)$, the belief $\Delta(\Br(n),\bot)$ is also confused.
Hence, by \autoref{lem-belief-basic}.4, $\ur(n) \bot$ is confused.
By \autoref{prop-feasible-no-confusion}, $\rho$~does not allow confusion.
So $\rho$~observes the $(n+1)$st letter.
Thus we have $\lp(n+1)=\lp(n)+1\le \lr(n)+1 = \lr(n+1)$, where the inequality is by the induction hypothesis.
Hence \eqref{eq-greedy-induction}~holds for $n+1$.
\end{enumerate}
\end{enumerate}
Thus we have shown~\eqref{eq-greedy-induction}.

Define $I \defeq \{(s,q) \mid (s,q) \text{ is not deciding}, \ \cras(s,q) = \infty\}$.
We will choose $K > |S|^2 \cdot |Q|^2$, so, by \autoref{lem-check-procr-k}, $\prho(K)$~coincides with $\prho(\infty)$ until possibly $\prho(K)$ encounters a pair $(s,q) \in I$.
Define:
\begin{align*}
 D \ &\defeq \ \{w \mid \prho(\infty) \text{ does not encounter any element of } I\}
\intertext{For each $(s,q) \in I$ define:}
 E(s,q) \ &\defeq \ \{w \mid (s,q) \text{ is the first element of~$I$ that $\prho(\infty)$ encounters on~$w$}\}
\end{align*}
So $D \cup \bigcup_{(s,q) \in I} E(s,q) = \Sigma^\omega$.
We show $\Ex(C_{\prho(K)}) \le \Ex(C_\rho)$ by conditioning separately on~$D$ and on each $E(s,q)$ that has positive probability.

There is, by \autoref{lem-pro-diagnoser}, almost surely a time, $n_\mathit{dec}$, when $\prho(K)$ encounters a deciding $(s,q)$. 
In the event~$D$, it follows $C_{\prho(K)} = \lp(n_\mathit{dec})$; i.e., we have almost surely
\begin{align}
 C_{\prho(K)} \ = \ \lp(n_\mathit{dec}) \ \mathop{\le}^\text{by~\eqref{eq-greedy-induction}} \ \lr(n_\mathit{dec}) \ \le \ C_\rho \qquad &\text{in the event~$D$.} \label{eq-event-D}
\intertext{%
Let $(s,q) \in I$ and consider the event $E(s,q)$.
Let us write $n$ for the time when $\prho(\infty)$, and thus $\prho(K)$, first encounters $(s,q)$.
Since $(s,q)$ is not deciding, the belief $\Bp(n)$ is not deciding.
Since $\prho(K)$ observes a letter at time~$n$, the belief $\Br(n) \succeq \Bp(n)$ is not deciding either, and $\rho$~needs to make at least one more observation.
So we have almost surely
}
 C_\rho - \lp(n) \ \ge \ \lr(n) + 1 - \lp(n) \ \mathop{\ge}^\text{by~\eqref{eq-greedy-induction}} \ 1 \qquad &\text{in the event~$E(s,q)$.} \label{eq-event-E}
\end{align}
\begin{enumerate}
\item
Suppose $\Ex(C_\rho - \lp(n) \mid E(s,q)) \le 1$.
Then, by~\eqref{eq-event-E}, we have $C_\rho = \lr(n) + 1$ almost surely in the event $E(s,q)$.
Thus, almost surely in $E(s,q)$, there exists $k \in \NN$ such that $\Delta(\Br, \bot^k a) \succeq \Delta((s,q), \bot^k a)$ is deciding for all~$a$.
It follows that there exists $k \in \NN$ such that $\Delta((s,q), \bot^k a)$ is deciding for all~$a$.
We will choose $K \ge k$, so $\Delta((s,q), \bot^K a)$ is deciding for all~$a$.
Hence, by~\eqref{eq-greedy-induction}, we have $C_{\prho(K)} = \lp(n) + 1 \le \lr(n) + 1 = C_\rho$ almost surely in the event~$E(s,q)$.
\item
Otherwise $\Ex(C_\rho - \lp(n) \mid E(s,q)) > 1$.
It follows from \autoref{lem-last-obs} that one can choose~$K$ large enough so that 
$\Ex(C_{\prho(K)} - \lp(n) \mid E(s,q)) < \Ex(C_\rho - \lp(n) \mid E(s,q))$.
We will choose $K$ in this way, so that we have $\Ex(C_{\prho(K)} \mid E(s,q)) < \Ex(C_\rho \mid E(s,q))$.
\end{enumerate}
By choosing~$K$ large enough so that for each $(s,q) \in I$ the respective constraint in item 1 or~2 is satisfied, we obtain $\Ex(C_{\prho(K)} \mid E(s,q)) \le \Ex(C_\rho \mid E(s,q))$ for each $(s,q) \in I$.
Combining this with~\eqref{eq-event-D} yields the result.
\end{proof}

}

\end{document}